\documentclass[11pt,a4paper]{article}
\usepackage{fullpage}
\usepackage{lmodern}
\usepackage[T1]{fontenc}
\usepackage[tbtags]{amsmath}
\allowdisplaybreaks
\usepackage{amssymb}
\usepackage{amsthm}
\usepackage{algorithm}
\usepackage[noend]{algpseudocode}
\usepackage{hyperref}
\usepackage[svgnames]{xcolor}
\usepackage[capitalise,nameinlink]{cleveref}
\hypersetup{colorlinks={true},linkcolor={DarkBlue},citecolor=[named]{DarkGreen}}

\algnewcommand\algorithmicinput{\textbf{INPUT:}}
\algnewcommand\INPUT{\item[\algorithmicinput]}
\algnewcommand\algorithmicoutput{\textbf{OUTPUT:}}
\algnewcommand\OUTPUT{\item[\algorithmicoutput]}

\theoremstyle{plain}
\newtheorem{theorem}{Theorem}
\newtheorem{lemma}{Lemma}
\def \R{\mathbb R}
\newcommand{\sset}[1]{\left\{ #1\right\}}
\newcommand{\ssets}[1]{\{ #1\}}
\newcommand{\fwh}[1]{\; \left| \; #1 \right.}
\newcommand{\fwhs}[1]{\; | \; #1 }
\newcommand{\card}[1]{\left| #1 \right|}

\DeclareMathOperator*{\argmin}{argmin}

\newcommand{\union}{\cup}
\newcommand{\map}{\longrightarrow}
\newcommand{\ifif}{\Longleftrightarrow}

\newcommand{\biglor}{\bigvee} 

\newcommand{\vecc}[1]{\ensuremath{\mathbf{#1}}}

\newcommand{\scsum}{C}

\newcommand{\poa}{\mathrm{PoA}}
\newcommand{\grat}[1]{\ensuremath{\Phi_{#1}}} 
\newcommand{\lambertW}[1]{\ensuremath{\mathcal{W}\left(#1\right)}} 
\newcommand{\lambertWnoarg}{\ensuremath{\mathcal{W}}}
\newcommand{\xe}[2]{x_{#1, \, e}(#2)}
\newcommand{\p}[3]{\varPhi_{#1}^{#3}(#2)}
\newcommand{\ce}[1]{c_e(#1)}
\newcommand{\pe}[1]{\phi_e(#1)}
\newcommand{\cu}[1]{C_u(#1)}
\newcommand{\cuj}[1]{C_{u_j}(#1)}
\newcommand{\pbar}[3]{\bar{\varPhi}_{#1}^{#3}(#2)}
\newcommand{\br}[1]{\mathcal{BR}_u(#1)}

\title{Computing Approximate Equilibria in Weighted Congestion Games via
Best-Responses\thanks{Supported by the Alexander von Humboldt Foundation with funds
from the German Federal Ministry of Education and Research (BMBF).\newline
\indent An abstract of this paper appeared in SAGT'20~\cite{gns2020_sagt}.
}}
\author{
		Yiannis Giannakopoulos\thanks{TU Munich. 
		Emails: 
		{\tt
		\{\href{mailto:yiannis.giannakopoulos@tum.de}{\nolinkurl{yiannis.giannakopoulos}},
		\href{mailto:andreas.s.schulz@tum.de}{\nolinkurl{andreas.s.schulz}}\}@tum.de}. 
		Associated researchers with the Research Training Group GRK
2201 ``Advanced Optimization in a Networked Economy'', funded by the German Research
Foundation (DFG).
		}
	\and
		Georgy Noarov\thanks{University of Pennsylvania. Email: 
			\href{mailto:gnoarov@seas.upenn.edu}{\nolinkurl{gnoarov@seas.upenn.edu}}. A significant part of this work was done while the second author was a visiting student at the Operations Research group of TU Munich.}
	\and 
		Andreas S.\ Schulz\footnotemark[2]
}
\date{November 24, 2020}

\begin{document}
\maketitle
\begin{abstract}
We present a deterministic polynomial-time algorithm for computing
$d^{d+o(d)}$-approximate (pure) Nash equilibria in (proportional sharing) weighted
congestion games with polynomial cost functions of degree at most~$d$. This is an
exponential improvement of the approximation factor with respect to the previously
best deterministic algorithm. An appealing additional feature of the algorithm is
that it only uses best-improvement steps in the actual game, as opposed to the
previously best algorithms, that first had to transform the game itself. Our
algorithm is an adaptation of the seminal algorithm by Caragiannis et
al.~\cite{Caragiannis2011,Caragiannis2015a}, but we utilize an approximate potential
function directly on the original game instead of an exact one on a modified game.

A critical component of our analysis, which is of independent interest, is the
derivation of a novel bound of $[d/\lambertW{d/\rho}]^{d+1}$ for the Price of
Anarchy (PoA) of $\rho$-approximate equilibria in weighted congestion games, where
$\lambertWnoarg$ is the Lambert-W function. More specifically, we show that this PoA
is \emph{exactly} equal to $\grat{d,\rho}^{d+1}$, where $\grat{d,\rho}$ is the
unique positive solution of the equation $\rho (x+1)^d=x^{d+1}$. Our upper bound is
derived via a smoothness-like argument, and thus holds even for mixed Nash and
correlated equilibria, while our lower bound is simple enough to apply even to
singleton congestion games.
\end{abstract}

\section{Introduction}
\emph{Congestion games} constitute one of the most important and well-studied class
of games in the field of \emph{algorithmic game
theory}~\cite{2007a,Roughgarden2005,Roughgarden2016}. These games are tailored to
model settings where selfish players compete over sets of common
resources.  Prominent examples include traffic routing in networks
and load balancing games (see, e.g., \cite[Chapters~18 and 20]{2007a}).
In these games' most general form, known as \emph{weighted} congestion games, each player has her own 
(positive) weight and the cost of a resource is a nondecreasing function of the
total weight of players using it.
An important special case is that of unweighted games, where all players have the
same weight. The cost of a resource then depends only on the number of players using
it.

Players are selfish and each one chooses a set of resources that 
minimizes her own cost. On the other hand, a central authority would aim at minimizing \emph{social cost}, that is, the sum of players' costs. It is well known that
these two objectives do not, in general, align: due to the selfish behaviour of 
players, the game may reach a stable state (i.e., a Nash
equilibrium~\cite{Karlin2017,Nash1951a}) that is suboptimal in terms of social
cost. This gap is formally captured by the fundamental notion of \emph{Price of
Anarchy} (PoA)~\cite{KP99}, defined as the ratio between the social cost of the worst equilibrium
and that of an optimal solution enforced by a centralized authority.
  
From the seminal work of Rosenthal~\cite{Rosenthal1973a}, we know that unweighted congestion
games always have \emph{(pure Nash) equilibria}\footnote{In this paper we focus
exclusively on \emph{pure} Nash equilibria; this is standard in the congestion games
literature.}. This is a direct consequence of the fact that they are \emph{potential
games}~\cite{Monderer:1996sp}. However, finding such a stable state is, in general,
computationally hard~\cite{Ackermann2008,Fabrikant2004a}. An important question
then is whether one can efficiently compute \emph{approximate}
equilibria~\cite{Caragiannis2011,Skopalik2008}. These are states of the game where
no player can unilaterally deviate and improve her cost by more than a factor of
$\rho \geq 1$; (exact) equilibria correspond to the special case where $\rho=1$.

The situation becomes even more challenging in the general setting of weighted
congestion games~\cite{Caragiannis2015a}, where exact equilibria may not even
exist~\cite{Goemans2005}. Consequently, weighted congestion games do not generally
admit a potential function. Thus, in this setting one needs more sophisticated
approaches and approximation tools to establish computability of approximate
equilibria. This is precisely the problem we study in the present paper: the
efficient computation of $\rho$-approximate equilibria in weighted congestion games,
with $\rho$ as small as possible. We focus on resource cost functions that are
polynomials (with nonnegative coefficients), parametrized by their degree $d$; this
is a common assumption in the literature of congestion games (see, e.g.,
\cite{Roughgarden2016,Aland2011,Christodoulou2015,cggs2018,Goemans2005}).

\subsection{Related Work} 
The potential function approach has long become a central tool for obtaining results
about existence and computability of approximate equilibria in weighted congestion
games. The concept of a potential function for unweighted congestion games was
proposed by Rosenthal~\cite{Rosenthal1973a}, who used it to prove
existence of pure Nash equilibria in such games. Later, Monderer and
Shapley~\cite{Monderer:1996sp} formally introduced and studied extensively the class
of potential games. As it turns out, \emph{weighted} congestion games do not admit
a potential function in general, even for ``well-behaved''
instances~\cite{Libman2001,Goemans2005,Fotakis2005a,Harks2012a}. Exceptions include
games with linear and exponential resource cost functions~\cite{Harks2012a}.
However, recently, Christodoulou et al.~\cite{cggs2018} showed that
polynomial weighted congestion games have an approximate analogue of a potential
function, which they called \emph{Faulhaber's} potential. An exact potential
function decreases whenever a player improves her cost, and even by the same amount.
In contrast, the approximate potential of~\cite{cggs2018} is only
guaranteed to decrease when a player deviates and improves her cost by a factor of
at least $\alpha$. Factor $\alpha$ is greater than $1$, but at most $d + 1$, where
$d$ is the degree of the game.  We will use this approximate potential function in
the analysis of our algorithm. Other approximate potential functions have been
successfully used before to establish the existence of approximate equilibria in
congestion games; see~\cite{Chen2008,Christodoulou2011a,Hansknecht2014}.
The current state-of-the-art is that congestion games of degree $d$ always admit
$d$-approximate equilibria, but there exist examples of such games that do not have
$\tilde{\varTheta}(\sqrt{d})$--approximate equilibria; these upper and lower bounds
are from the very recent work of Caragiannis and Fanelli~\cite{Caragiannis:2019aa}
and Christodoulou et al.~\cite{cggpw2020}, respectively.

On the algorithmic side, there have been many negative results concerning exact
equilibria in various classes of congestion games. Fabrikant et
al.~\cite{Fabrikant2004a} showed that even in the unweighted case of a network
congestion game, computing equilibria is PLS-complete. Dunkel and
Schulz~\cite{Dunkel2008} showed that it is strongly NP-complete to determine whether
an equilibrium exists in a given weighted congestion game.\footnote{Christodoulou et al.~\cite{cggpw2020} showed that this remains hard even for
approximate equilibria and polynomial congestion games.} As a further negative
result, Ackermann et al.~\cite{Ackermann2008} proved that it is PLS-complete to
compute equilibria even in the linear unweighted case. These hardness results
motivated the search for polynomial-time methods to compute \emph{approximate} equilibria.
In general, this remains a computationally hard problem; Skopalik and
V\"ocking~\cite{Skopalik2008} showed that for any polynomially computable $\rho$,
finding a $\rho$-approximate equilibrium in a congestion game is a PLS-complete
problem. The focus of research thus shifted towards searching for positive
algorithmic results for $\rho$-approximate equilibria of various special classes of
weighted congestion games. The first such result was obtained by Chien and
Sinclair~\cite{Chien2011}, who showed convergence of the best-response dynamics, in
symmetric unweighted congestion games with ``well-behaved'' cost functions, to $(1 +
\varepsilon)$-approximate equilibria in polynomially many steps with respect to
$\varepsilon^{-1}$ and the number of players.

The next significant positive result of this kind, and of great importance for our work, was obtained for approximate
equilibria in polynomial unweighted congestion games by Caragiannis et al.~\cite{Caragiannis2011}.
They designed a seminal deterministic algorithm that efficiently computes $d^{O(d)}$-approximate equilibria in such
games. Subsequently, Caragiannis et al.~\cite{Caragiannis2015a} extended this algorithm to handle the weighted
case, achieving an approximation factor of $d^{2d + o(d)}$. 
In the present paper, we reduce this
factor to $d^{d + o(d)}$. The algorithm in \cite{Caragiannis2015a} first transforms
the original game into an \emph{approximating game} (called $\Psi$-game) defined over the same
players and states. 
Then it finds and returns a state which is a $d^{d + o(d)}$-approximate equilibrium
of this new game. Caragiannis et al.~\cite{Caragiannis2015a} show that when
translated back to the original game, the approximation guarantee can only
deteriorate by a factor of $d!$, giving their $d^{2d + o(d)}$-approximation result.

It is important to note here that the sequence of moves chosen by Algorithm~1
in~\cite{Caragiannis2015a} need not be a ``real'' best-response sequence when
projected back into the original game. As a matter of fact, it may contain moves that increase the deviating
player's individual cost. 
In an attempt to address this,
Caragiannis et al.~\cite{Caragiannis2015a} themselves present also a modification of Algorithm~1. Their Algorithm~2~\cite[Section~5]{Caragiannis2015a} actually runs in the original game but unfortunately offers a significantly worse approximation guarantee of $d^{O(d^2)}$.

Very relevant to ours is also the work of Feldotto et al.~\cite{Feldotto2017} who
describe a \emph{randomized} variant of \cite{Caragiannis2015a}'s algorithm which is
able to compute $d^{d + o(d)}$-approximate equilibria in weighted congestion games
(but only with high probability); this is of the same order as the approximation
ratio of our \emph{deterministic} algorithm.\footnote{This is not a mere
coincidence: both in~\cite{Feldotto2017} and our paper (but, as a matter of fact,
in~\cite{Caragiannis2015a} as well), it turns out that the ``bottleneck'' in the
approximation ratio bound is essentially given by PoA-style bounds (see, e.g., our
``$\rho$-PoA'' bound in~\cref{th:poa_upper} and the``$\rho$-stretch'' notion
in~\cite{Feldotto2017,Caragiannis2015a}) which are roughly of the same $d^d$ order.}
Similar to \cite{Caragiannis2015a} though, their algorithm does not perform actual
best-improvement moves in the original game: it runs in a modified instance where the Shapley value rule of~\cite{Kollias2015} is used to share the total cost of a
resource among the players occupying it. Under this cost sharing rule, however, it is computationally hard to identify players'
costs. Instead, one has to
estimate them using a sampling approach, which is exactly the source of
randomization in \cite{Feldotto2017}'s algorithm.
Another subtle limitation of the algorithm in~\cite{Feldotto2017} is that its
running time is polynomial in the size of the strategy sets. Although this is
absolutely fine when the input game is given explicitly, it can have serious
implications for games that are \emph{succinctly representable} but have
exponentially many strategies; an important canonical example are network congestion
games.

In another paper, Feldotto et al.~\cite{Feldotto2014} explored
PoA-like bounds on the potential function in unweighted games. This enabled them to
bound from above the approximation factor that the algorithmic framework of
\cite{Caragiannis2011} yields when applied to unweighted games with general cost
functions.

The study of the Price of Anarchy ($\poa$) was initiated by Koutsoupias and Papadimitriou~\cite{KP99}. One of the
first significant results concerning tight bounds on the $\poa$ of atomic congestion
games was obtained by Awerbuch et al.~\cite{Awerbuch2013} and Christodoulou and Koutsoupias~\cite{CK05a}. They proved the tight bound of $5/2$ on the
exact $\poa$ of linear unweighted congestion games. In the next few years, several
$\poa$ results were obtained. For example, Gairing and Schoppmann~\cite{Gairing2007} provided various
upper and lower bounds for the exact $\poa$ of
singleton\footnote{See~\cref{foot:def_singleton} for a formal definition.}
unweighted congestion games. Subsequently, Aland et al.~\cite{Aland2011} introduced a systematic
approach to upper-bounding the exact $\poa$ of polynomial weighted congestion games,
which was later extended to general classes of cost functions and named
\emph{smoothness framework} in \cite{Bhawalkar2014a,Roughgarden2015}.
Aland et al.~\cite{Aland2011} gave the tight bound of $\grat{d}^{d+1}$ on the $\poa$ of exact
pure NE in polynomial weighted games, where $\grat{d}$ is the unique root of the
equation $(x + 1)^{d} = x^{d+1}$.
Based on the same technique, Christodoulou et al.~\cite{Christodoulou2011a} provided a
tight bound on the PoA of $\rho$-approximate equilibria in unweighted congestion
games. It turned out to be equal to $\frac{\rho((z + 1)^{2d+1} -
z^{d+1}(z+2)^{d})}{(z+1)^{d+1} - z^{d+1} - \rho((z+2)^{d} - (z+1)^{d})}$, where $z$
is the maximum integer that satisfies $\frac{z^{d+1}}{(z+1)^{d}} < \rho$. In our
notation, this is equivalent to $z = \left \lfloor \grat{d, \rho} \right\rfloor $.  

Since the development of the smoothness method, other approaches to finding tight
bounds on the $\poa$ have been investigated. Recently, Bil\`o~\cite{Bilo2017} was able to
rederive, through the use of a primal-dual framework, the upper bound on the $\poa$
of linear unweighted games from \cite{CK05a}.  He also provided a simplified lower
bound instance. Furthermore, he was able to show the upper bound of
$\left(\frac{\rho + \sqrt{\rho^2 + 4\rho}}{2} \right)^{2}$ on the $\poa$ of
$\rho$-approximate equilibria for the special case of linear weighted games. It
turns out to be equal to $\grat{1, \rho}^2$, the special case for $d = 1$ of the
general tight bound $\grat{d, \rho}^{d+1}$ that we present in this paper. Moreover, he
provided matching instances with $\poa$ equal to $\grat{1, \rho}^2$, for $\rho$ in a
certain subset of $[1, \infty)$.

\subsection{Our Results and Techniques} \label{subsec:ourres}

We study approximate (pure Nash) equilibria in polynomial weighted congestion games of
degree~$d$. Our main result is a polynomial-time deterministic algorithm for computing $d^{d +
o(d)}$-approximate equilibria in such games.
Our algorithm runs in polynomial time in the description of the game, even for
succinctly representable games with exponentially large strategy spaces; in
particular, it is applicable to \emph{network} congestion games.\footnote{We want to
note here that the algorithms of Caragiannis et al.~\cite{Caragiannis2015a} can also
handle network games without computational issues. However, as we already mentioned
above, that is not the case for the algorithm in~\cite{Feldotto2017}.}
This result improves upon the $d^{2d+o(d)}$-approximation factor of the seminal algorithm of Caragiannis et al.~\cite{Caragiannis2015a,Caragiannis2011}. Interestingly, our algorithm can also be readily used to \emph{deterministically}
compute $d^{d + o(d)}$-approximate equilibria in weighted congestion games under
\emph{Shapley cost sharing}; this is of the same order as the \emph{randomized}
$d^{d + o(d)}$ guarantee of Feldotto et al.~\cite{Feldotto2017}.\footnote{This is an
immediate consequence of our deterministic $d^{d + o(d)}$ guarantee for the
proportional sharing setting studied here, paired with Lemma~8 in~\cite{Feldotto2017}.}

Our algorithm, as well as the outline of its analysis, is clearly based on the ideas
in \cite{Caragiannis2015a}. However, there is a fundamental difference between our
approach and previous ones. Our algorithm builds a polynomially-long sequence of
best-response moves (leading from any state to an approximate equilibrium) \emph{in the actual game}; we then utilize, in the absence of an
exact potential function for the game, an \emph{approximate potential} introduced
by Christodoulou et al.~\cite{cggs2018} in order to analyse the runtime and
approximation guarantee.
By contrast, the original $d^{2d+o(d)}$-approximation algorithm of Caragiannis et al.~\cite{Caragiannis2015a}, as well as the randomized $d^{d+o(d)}$-approximation algorithm of Feldotto et al.~\cite{Feldotto2017}, first transforms the input game 
into an \emph{exact-potential} auxiliary game and performs best-responses \emph{in that modified game}; the outcome is then projected back into the original game, at the expense of a certain increase in the approximation factor.\footnote{This increase is by a factor of $d!$ for Algorithm~1 in \cite{Caragiannis2015a}, and $O(d^2)$ for \cite{Feldotto2017}.}
A noteworthy exception is Algorithm~2
of~\cite[Section~5]{Caragiannis2015a}, which also creates a best-response sequence
in the original game; this, however, is achieved at the expense of a significant deterioration of the
approximation guarantee of the computed equilibrium, which becomes
$d^{O(d^2)}$. Thus, when compared to Algorithm~2 of~\cite{Caragiannis2015a}, the improvement of the approximation factor provided by our $d^{d+o(d)}$-approximate algorithm is even more significant.

We believe that our approach offers a number of advantages compared to the
previously best algorithms (namely, Algorithm~1 of~\cite{Caragiannis2015a} and the
algorithm of~\cite{Feldotto2017}), at the cost of our proofs being arguably more
involved due to the use of an approximate potential to the main game, instead of an
exact potential to an auxiliary game. First, similarly to Algorithm~2 of~\cite{Caragiannis2015a}, since our algorithm runs on the actual game, it can be
interpreted as a more natural \emph{learning} process via ``real'' Nash dynamics.
An additional implication from a computational complexity perspective is that our algorithm can be applied even on games for which we are only given a best-response \emph{oracle}.
Moreover, since also our analysis is performed in a direct way on the original game (in contrast to the analyses of Algorithms~1 and 2 of~\cite{Caragiannis2015a} and the analysis of~\cite{Feldotto2017}, each of which considers an auxiliary game), we believe that it provides a more transparent understanding of the inner workings of the original algorithmic paradigm of Caragiannis et al.~\cite{Caragiannis2015a}. 

As a necessary tool for proving the approximation guarantee of the algorithm and as
a result of independent interest, we obtain a tight bound on the PoA of
$\rho$-approximate equilibria, denoted by $\poa_d(\rho)$, of polynomial weighted
congestion games of degree $d$, for any $\rho \geq 1$ and degree $d \geq 1$. It
turns out to be equal to $\grat{d, \rho}^{d+1}$, where $\grat{d, \rho}$ is the
unique positive solution of the equation $\rho (x+1)^d=x^{d+1}$. This bound
generalizes the following results: the tight bound of $\grat{d}^{d+1}$ on the $\poa$
(of exact equilibria) of weighted congestion games \cite{Aland2011}; the tight bound
on the $\rho$-approximate $\poa$ of \emph{unweighted} congestion games
\cite{Christodoulou2011a}; and the upper bound on the $\rho$-approximate $\poa$ of
\emph{linear} weighted congestion games \cite{Bilo2017}.
Our matching lower bound proof extends an example from \cite{Gairing2007} that
bounds the $\poa$ of singleton unweighted congestion games. As such, our lower bound
is easily verified to hold for singleton and network weighted congestion
games.\footnote{The fact that the worst-case PoA can be realized at such simple
singleton games should come as no surprise, due to the work
of Bil\`o and Vinci~\cite[Theorem~1]{Bilo2017a}. Our contribution here lies in determining the
actual value of the PoA.}
To prove the upper bound, we essentially utilize the \emph{smoothness method}
developed in \cite{Aland2011,Bhawalkar2014a,Roughgarden2015}. The smoothness
approach automatically extends the validity of our tight $\poa$ bound from pure Nash
to mixed Nash and correlated equilibria as well
\cite{Bhawalkar2014a,Roughgarden2015}.

One further contribution is an analytic upper bound $\poa_d(\rho) \leq
\left[\frac{d}{\lambertW{d/\rho}}\right]^{d+1}$, involving the Lambert-W function. 
This bound adds to the understanding of the different asymptotic behaviour of
$\poa_d(\rho)$ with respect to each of the two parameters, $d$ and $\rho$, and plays
an important role in deriving the desired approximation factor of $d^{d+o(d)}$ in
the analysis of the main algorithm of our paper. It is interesting to note here that
this bound also generalizes, in a smooth way with respect to $\rho$, a similar
result presented in~\cite{cggs2018} for the special case of exact
equilibria (i.e., $\rho=1$).

All proofs omitted from the main text can be found in the Appendix.

\section{Model and Notation}
\label{sec:notation}
We denote by $\R$ and $\R_{\geq 0}$ the set of real and nonnegative real numbers, respectively, and by $\mathfrak{P}_d$ the class of polynomials of degree at most $d$ with
nonnegative coefficients\footnote{Formally,
$\mathfrak{P}_d=\ssets{f:\R\map\R\fwhs{f(x)=\sum_{i=0}^da_i
x^i,\;\;\text{where}\;a_i\in\R_{\geq 0}\;\text{for all}\;i}}$.}.
A well-established notation in the literature of congestion games is that of $\Phi_{d}$, for
$d$ a positive integer, as the unique positive root of the equation $(x+1)^d=x^{d+1}$. Notice
how, the special case of $d=1$ corresponds to the \emph{golden ratio} constant
$\phi\approx 1.618$. In this paper, we introduce a further generalization by
defining, for all $\rho \geq 1$,  $\grat{d,\rho}$ to be the unique positive root of
the equation $\rho (x+1)^d=x^{d+1}$.
Also, we shall make use of (the principal real branch
of) the classical function known as the \emph{Lambert-W
function}~\cite{Corless1996}: for $\tau \geq 0$, $\lambertW{\tau}$ is defined to be
the unique solution to the equation $x\cdot e^{x} = \tau$.

A \emph{polynomial\footnote{We shall usually omit the word ``polynomial'' and refer
to these games as weighted congestion games of degree $d$, or simply as congestion
games, when this causes no confusion.} (weighted) congestion game of degree $d$},
with $d$ a positive integer, is a tuple $\Gamma=(N, E, (w_u)_{u \in N}, (S_u)_{u
\in N}, (c_e)_{e
\in E})$. Here, $N$ is a (finite) set of $\card{N}=n$ \emph{players} and $E$ is a
finite set of \emph{resources}. Each resource $e \in E$ has a
polynomial \emph{cost function} $c_e\in \mathfrak{P}_d$.
Every player $u \in N$ has a set of \emph{strategies} $S_u\subseteq 2^{E}$ and
each vector $\vecc s\in S:= \times_{u \in N} S_u$ will be called a
\emph{state} (or \emph{strategy profile}) of the game $\Gamma$.
Following standard game-theoretic notation, for any $\vecc s\in S$ and $u \in N$, we
denote by $\vecc s_{-u}$ the profile of strategies of all players if we remove the strategy $s_u$ of player $u$;
in this way, we have $\vecc s=(s_u,\vecc s_{-u})$.
Finally, each player $u \in N$ has a real positive \emph{weight} $w_u >0$. However, we
may henceforth assume that $w_u \geq 1$ for all $u \in N$, as we can without loss of generality appropriately scale
player weights and cost functions, without affecting our results in
this paper.

Given $\vecc s\in S$, we let $x_e(\vecc s):= \sum_{u:e\in s_u} w_e$ denote the total weight of players using
resource $e$ in state $\vecc s$. 
Generalizing this definition to any \emph{group} of
players $R \subseteq N$, we let $\xe{R}{\vecc s}:= \sum_{u\in R:e\in s_u} w_e $. 
The cost of a player $u$ at
state $\vecc s$ is defined as 
$$\cu{\vecc s} := w_u \sum_{e \in s_u} \ce{x_e(\vecc s)}.$$ 
Players are selfish and rational, and thus choose strategies as to minimize their
own cost. Let $\br{\vecc s}=\br{\vecc s_{-u}}$ be a \emph{best-response} strategy of player
$u$ to the strategies $\vecc s_{-u}$ of the other players, that is, $\br{\vecc s} \in
\argmin_{s_u' \in S_u} \cu{\vecc s_{-u}, s_u'}$ (in case of ties, we make an arbitrary selection).
A state $\vecc s$ of the game is a \emph{(pure Nash) equilibrium}, if all players are already playing best-responses, that is, no player can unilaterally improve her costs; formally, $\cu{\vecc s} \leq \cu{s_u',\vecc s_{-u}}$ for all $u\in N$ and $s_u'\in S_u$.

For a real parameter $\rho \geq 1$, a unilateral deviation of player $u$ to strategy
$s_u'$ from state $\vecc s$ is called a \emph{$\rho$-move} if $\cu{\vecc s} > \rho\cdot
\cu{\vecc s_{-u}, s_u'}$.
Extending the notion of an equilibrium in two directions, we call a state $\vecc s$
a \emph{$\rho$-approximate equilibrium} (or simply a $\rho$-equilibrium) for a given
group of players $R \subseteq N$, if none of the players in $R$ has a $\rho$-move;
formally, $\cu{\vecc s} \leq \rho \cu{s_u',\vecc s_{-u}}$ for all $u\in R$ and
$s_u'\in S_u$. If this holds for $R=N$, then we simply refer to $\vecc s$ as a
$\rho$-equilibrium of our game. We use $\mathcal{Q}_\rho^\Gamma$ to denote the set
of all $\rho$-equilibria of game $\Gamma$.

Ideally, our objective is to find states that induce low total cost in our game; we
capture this notion by defining the \emph{social cost} $\scsum(\vecc s)$ of a
state $\vecc s$ to be the sum of the players' costs, i.e., $\scsum(\vecc s):=\sum_{u
\in N} \cu{\vecc s}$. Extending this to any subset of players $R\subseteq N$, we
also denote 
$$\scsum_{R}(\vecc s)
:= \sum_{u \in R} \cu{\vecc s}
=\sum_{u \in R}w_u\sum_{e\in s_u}\ce{x_e(\vecc s)}
=\sum_{e\in E}\xe{R}{\vecc s} \ce{x_e(\vecc s)}.$$
Clearly, $\scsum_{N}(\vecc s)=\scsum(\vecc s)$.

The standard way to quantify the inefficiency due to selfish behaviour, is to study
the worst-case ratio between any equilibrium and the optimal solution, quantified by
the notion of \emph{the Price of Anarchy (PoA)}. Formally, given a game $\Gamma$ and
a parameter $\rho\geq 1$, the PoA of $\rho$-equilibria (or simply the $\rho$-PoA) of
$\Gamma$ is $\poa(\Gamma) := \max_{\vecc s\in\mathcal{Q}^\Gamma_{\rho}}
\frac{\scsum(\vecc s)}{\scsum(\vecc s^{*})}$, where $\vecc s^*\in\argmin_{\vecc s}
C(\vecc s)$. Finally, taking the worst case over all polynomial congestion games of
degree $d$, we can define the $\rho$-PoA of degree $d$ as
$$
\poa_d(\rho) := \sup_{\Gamma} \poa(\Gamma)=\sup_{\Gamma} \max_{\vecc s\in\mathcal{Q}^\Gamma_{\rho}} \frac{\scsum(\vecc s)}{\scsum(\vecc s^{*})}.
$$

\section{The Price of Anarchy} \label{sec:poa}
In this section we present our tight bound on the PoA of $\rho$-approximate
equilibria for (weighted) congestion games. We first extend the smoothness method
of Aland et al.~\cite{Aland2011} to obtain the upper bound on the $\poa$ (\cref{th:poa_upper}),
and then explicitly construct an example that extends a result
of Gairing and Schoppmann~\cite{Gairing2007} and provides the matching lower bound on the $\poa$
(\cref{th:poa_lower}).
We note here that there is a specific reason that this section precedes
\cref{sec:algorithm}, where our algorithm for computing approximate pure Nash
equilibria is presented. The estimation of the approximation guarantee of the
algorithm requires the use of the closed-form bound on the $\poa_d(\rho)$ we provide
in~\cref{th:poa_upper} and, furthermore, the ``Key Property'' of our algorithm
(\cref{thm}) rests critically on an application of \cref{partialpoa} below.

\subsection{Upper Bound} \label{subsec:poa_upper}

We formulate our upper bound on the PoA as the following theorem:
\begin{theorem}
\label{th:poa_upper} The Price of Anarchy of $\rho$-approximate equilibria in
(weighted) polynomial congestion games of degree $d$, is at most
$\grat{d,\rho}^{d+1}$, where $\grat{d,\rho}$ is the unique positive root of the
equation $\rho(x+1)^d=x^{d+1}$. In particular,
$$
\poa_d(\rho) \leq \left[\frac{d}{\lambertW{d/\rho}}\right]^{d+1},
$$
where $\lambertW{\cdot}$ denotes the Lambert-W function.
\end{theorem}

\Cref{th:poa_upper} is a direct consequence of the following \cref{partialpoa},
applied with $R =N$, and \cref{lem:phi_lambert}. The reason we are proving a more
general version of~\cref{partialpoa} than what's needed for just establishing our
PoA upper bound of~\Cref{th:poa_upper}, is that we will actually need it for the
analysis of our main algorithm in~\cref{subseq:runtime}.

\begin{lemma} \label{partialpoa} 
For any group of players $R \subseteq N$, let
$\vecc s$ and $\vecc s^{*}$ be states such that $\vecc s$ is a $\rho$-equilibrium
for group $R$ and every player in $N \setminus R$ uses the same strategy in both
$\vecc s$ and $\vecc s^{*}$. Then, the social cost ratio of the two states is
bounded by $\frac{\scsum_R({\vecc s})}{\scsum_R({\vecc s}^{*})} \leq
\grat{d,\rho}^{d+1}$.
\end{lemma}

The proof of \cref{partialpoa} will essentially follow the smoothness
technique~\cite{Roughgarden2015} (see, e.g., \cite[Theorem~14.6]{Roughgarden2016}).
However, special care still needs to be taken related to the fact that only a subset
$R$ of players is deviating between the two states $\vecc s$ and $\vecc s^*$. 
In particular, the key step in the
smoothness derivation is captured by the following lemma (proved in~\cref{app:poa})
that quantifies the PoA bound:
\begin{lemma} 
\label{mainlemma} 
For any constant $\rho \geq 1$ and positive integer $d$,
\begin{align*}
B :=  
\mkern-12mu
\inf_{\substack{\lambda \in \R \\ \mu \in (0, \frac{1}{\rho})}}
\mkern-12mu
\sset{\left.\frac{\lambda \rho}{1 - \mu \rho} \right| \forall x, y, z \geq 0, f 
\in  \mathfrak{P}_d  :  y  f(z + x + y) \leq \lambda y  f(z+ y) + \mu x  f(z + x)} 
= \grat{d,\rho}^{d+1}.
\end{align*}
\end{lemma}
The constraint that~\cref{mainlemma} imposes on parameters $\lambda$ and $\mu$ is
slightly more general than the analogous lemma in the smoothness derivation
of Aland et al.~\cite{Aland2011}; namely, our condition contains an extra variable $z$. This is
a consequence of exactly the aforementioned fact that~\cref{mainlemma} is tailored
to upper bounding a generalization of the PoA for \emph{groups} of players.
\begin{proof}[Proof of~\cref{partialpoa}]
Assume that $\lambda \in \R$ and $\mu \in (0,\frac{1}{\rho})$ are parameters such
that, for any polynomial $f$ of degree $d$ with nonnegative coefficients and for any
$x, y, z \geq 0$, it is
$$
y f(z + x + y) \leq \lambda y f(z + y) + \mu x f(z + x).
$$
Applying this for the cost function $\ce{\cdot}$ of any resource $e$, and replacing
$x \gets \xe{R}{{\vecc s}}$, $y \gets \xe{R}{{\vecc s}^{*}}$, and $z
\gets \xe{N \setminus R}{{\vecc s}} = \xe{N \setminus R}{{\vecc s}^{*}}$ (the last equality holding
due to the fact that every player in $N \setminus R$ uses the same strategy in ${\vecc s}$
and ${\vecc s}^{*})$ we have that
\begin{equation} \label{ineqSmoothness}
\xe{R}{{\vecc s}^{*}} \, \ce{ x_e({\vecc s}) + \xe{R}{{\vecc s}^{*}} }  \leq  \lambda \,
\xe{R}{{\vecc s}^{*}} \, \ce{ x_e({\vecc s}^{*}) } + \mu \, \xe{R}{{\vecc s}} \, \ce{ x_e({\vecc s}) }.
\end{equation} 
Here we also used that  $z + x + y = x_e({\vecc s}) + \xe{R}{{\vecc s}^{*}}$, $z + y
= x_e({\vecc s}^{*})$, and $z + x = x_e({\vecc s})$. Summing \eqref{ineqSmoothness} over all
resources $e$, we obtain the following inequality:
\begin{equation} \label{ineqmu}
\sum_{e\in E} \xe{R}{{\vecc s}^{*}} \, \ce{ x_e({\vecc s}) + \xe{R}{{\vecc s}^{*}} }  \leq  \lambda \,
\scsum_{R}({\vecc s}^{*})  +  \mu \, \scsum_{R}({\vecc s}).
\end{equation} 

Next, using the fact that $\vecc s$ is a $\rho$-equilibrium we can upper-bound the
social cost of the players in $R$ by
\begin{equation} \label{eq2}
\scsum_R(\vecc s)
=\sum_{u \in R} \cu{{\vecc s}} 
\leq \rho\sum_{u \in R}\cu{{\vecc s}_{-u}, s^{*}_u} 
 =\rho\sum_{u \in R}w_u\sum_{e\in s^{*}_u} \ce{x_e({\vecc s}_{-u}, s^{*}_u)}.
\end{equation}
Now, observe that for any player $u\in R$ and any resource $e\in s^*_u$ that player
$u$ uses in profile $\vecc s^*$, it is
$$
x_e(\vecc s_{-u}, s^*_u) \leq x_e(\vecc s)+w_u \leq x_e({\vecc s}) +\xe{R}{{\vecc s}^{*}}.
$$
The first inequality holds because player $u$ is the only one deviating between
states $\vecc s$ and $(\vecc s_{-u}, s^*_u)$, while the second one because $u$
definitely uses resource $e$ in profile $s^*$. Using the above, due to the
monotonicity of the cost functions $c_e$, the bound in~\eqref{eq2} can be further
developed to give us
$$
\scsum_R(\vecc s) 
\leq \rho\sum_{u \in R}w_u\sum_{e\in s^{*}_u} \ce{x_e({\vecc s}) +\xe{R}{{\vecc s}^{*}}}
= \rho \sum_{e\in E} \xe{R}{{\vecc s}^{*}} \, \ce{ x_e({\vecc s}) + \xe{R}{{\vecc s}^{*}} }
$$ 
and thus, deploying the bound from~\eqref{ineqmu}, we finally arrive at
$$
C_R({\vecc s}) \leq \rho \lambda \, C_R({\vecc s}^{*})  +  \rho\mu \,C_R({\vecc s}),
$$
which is equivalent to
$$
\frac{C_R({\vecc s})}{C_R({\vecc s}^{*})} \leq \frac{\lambda \rho}{1 -
\mu \rho}.
$$
Taking the infimum of the right-hand side, over the set of all feasible parameters
$\lambda\in\R$ and $\mu\in (0,\frac{1}{\rho})$, \cref{mainlemma} gives us desired
upper bound of $\frac{\scsum_R({\vecc s})}{\scsum_R({\vecc s}^{*})} \leq \grat{d,\rho}^{d+1}$.
\end{proof}

We conclude this section by presenting the following useful bound on the generalized
golden ratio $\grat{d,\rho}$, which is used in \cref{th:poa_upper} to get the
corresponding analytic expression for our $\poa$ bound. As discussed in the
introduction of the current section, we will use it in the proof of~\cref{equil}
in~\cref{sec:algorithm}, for deriving the improved approximation guarantee of our
algorithm.

\begin{lemma}
\label{lem:phi_lambert}
For any $\rho\geq 1$ and any positive integer $d$,
$$
\grat{d,\rho} \leq \frac{d}{\lambertW{d/\rho}},
$$
where $\lambertW{\cdot}$ is the Lambert-W function.
\end{lemma}

\begin{proof}
Recall that $\grat{d,\rho}$ is the solution of equation $\rho (x+1)^d=x^{d+1}$,
which can be rewritten as 
$$\frac{d(x+1)^d}{x^{d+1}} = \frac{d}{\rho}.$$ 
By the proof of~\cref{maximumlemma} (see function $h$), this equation has a
unique positive root and the left side is monotonically decreasing as a function of
$x$; thus, to conclude the proof of our lemma, it suffices to
prove that
\begin{equation} 
\label{toprove}
\frac{d(\tilde{x}+1)^d}{\tilde{x}^{d+1}} \leq \frac{d}{\rho}
\qquad\qquad 
\text{for}\quad \tilde{x} := \frac{d}{\lambertW{\frac{d}{\rho}}}.
\end{equation}
Indeed, substituting for convenience $\tilde{y} := \frac{d}{\tilde{x}}=\lambertW{\frac{d}{\rho}}$, we have that
$$
\frac{d(\tilde{x}+1)^d}{\tilde{x}^{d+1}} 
= \tilde{y} \left(1+\frac{\tilde{y}}{d} \right)^d
= \tilde{y} \left[\left(1+\frac{\tilde{y}}{d} \right)^\frac{d}{\tilde y}\right]^{\tilde y}
\leq  \tilde{y} e^{\tilde y}
=\lambertW{\frac{d}{\rho}}
e^{\lambertW{\frac{d}{\rho}}} 
= \frac{d}{\rho}.
$$
For the inequality we used the fact that $\left(1+\frac{1}{t}\right)^t<e$ for all $t>0$. The last equality is a direct consequence of the definition of the Lambert-W function.
\end{proof}

\subsection{Lower Bound} \label{subsec:poa_lower} 
To prove a matching PoA lower bound to the upper bound of~\cref{subsec:poa_upper},
we consider a simple instance involving $n$ players and $n + 1$ resources. Each
player has just $2$ strategies, and each strategy consists of a single resource;
letting $n \to \infty$, we obtain the desired lower bound of $\grat{d,\rho}^{d+1}$.
This bound extends smoothly the lower bound of $\grat{d}^{d+1}$ for the PoA of exact
($\rho=1$) equilibria by Gairing and Schoppmann~\cite[Theorem~4]{Gairing2007}. We
also want to mention here that the construction used in the proof of
\cref{th:poa_lower} below can be extended to apply to network congestion games (see,
e.g., \cite[Proposition~3.4]{cggs2018}).

\begin{theorem} \label{lowerbound}
\label{th:poa_lower}
Let $d$ be a positive integer and $\rho\geq 1$. For every $\varepsilon > 0$, there
exists a (singleton\footnote{\label{foot:def_singleton}In \emph{singleton}
congestion games the strategies of all players consist of a single resource.
Formally $\card{s_u}=1$, for any $u\in N$ and all $s_u\in S_u$.}) weighted
polynomial congestion game of degree $d$, whose $\rho$-approximate $\poa$ is at least
$\grat{d,\rho}^{d+1} - \varepsilon$, where $\grat{d,\rho}$ is the unique positive
root of the equation $\rho(x+1)^d=x^{d+1}$.
\end{theorem}

\begin{proof}
Consider the following congestion game, with $n$ players $N=\{1,\dots,n\}$ and $n+1$
resources $E= \{1, \ldots, n+1\}$. Each player $i$ has a weight of $w_i=w^i$, where
$w := \frac{1}{\grat{d, \rho}}$. Resources have cost functions
$$
c_{j}(t)=
\begin{cases}
\frac{1}{\rho}\grat{d, \rho}^{d+2}, &j=1,\\[10pt]
\grat{d, \rho}^{(d+1)j} t^d, &j=2,\dots,n+1.
\end{cases}
$$
Each player $i$ has only two available strategies, denoted by $s_i^*$ and $s_i$:
either use only resource $i$, or only resource $i + 1$. Formally,
$S_i=\ssets{s_i^*,s_i}$, where $s_i^*=\ssets{i}$ and $s_i=\ssets{i+1}$.
The social cost of profile $\vecc s=(s_1,\dots,s_n)$, where every player $i$ uses the $(i+1)$-th resource, is
\begin{align*}
\scsum({\vecc s}) 
& = \sum_{i = 1}^{n} C_{i}({\vecc s}) 
= \sum_{i = 1}^{n} w_i c_{i+1}(w_i)\\ 
& = \sum_{i = 1}^{n} w^{i}
\grat{d, \rho}^{(d+1)(i+1)}   w^{i d} 
= \sum_{i = 1}^{n}  \left(\grat{d,
\rho}^{i+1} \, w^{i} \right)^{d+1}\\ 
& = \sum_{i = 1}^{n}  \grat{d,
\rho}^{d+1} = n \, \grat{d, \rho}^{d+1}.
\end{align*}
while that of $\vecc s^*=(s_1^*,\dots,s_n^*)$, where player $i$ uses the $i$-th resource is
\begin{align*}
\scsum({\vecc s}^{*}) 
= \sum_{i = 1}^{n} w_i c_{i}(w_i) 
= w \frac{1}{\rho}\grat{d,
\rho}^{d+2} + \sum\limits_{i = 2}^{n} w^{i} \grat{d, \rho}^{(d+1)i}
w^{id} 
= \frac{1}{\rho}\grat{d, \rho}^{d+1} + \sum\limits_{i = 2}^{n} 1
= \frac{1}{\rho}\grat{d, \rho}^{d+1} + n - 1.
\end{align*}

We now claim that profile $\vecc s$ is a $\rho$-equilibrium. Indeed, for player $i=1$,
$$
\frac{C_{1}({\vecc s})}{C_{1}({\vecc s}_{-1}, s_1^{*})} 
=\frac{c_{2}(w_1)}{c_{1}(w_1)}
= \frac{\grat{d,
\rho}^{2(d+1)} w^d}{\frac{1}{\rho}\grat{d, \rho}^{d+2}} =\rho \frac{\grat{d,
\rho}^{2(d+1)-d}}{\grat{d, \rho}^{d + 2}} = \rho,
$$
and for all players $i=2,\dots,n$,
$$\frac{C_{i}({\vecc s})}{C_{i}({\vecc s}_{-u_{i}}, s_i^{*})}
=\frac{c_{i+1}(w_i)}{c_{i}(w_i + w_{i-1})}
=
\frac{\grat{d,
\rho}^{(d+1)(i+1)} (w^{i})^d}{\grat{d, \rho}^{(d+1)i} (w^{i} +
w^{i-1})^{d}} 
= \grat{d, \rho}^{d+1} \left(1 + \frac{1}{w} \right)^{-d} = \grat{d,
\rho}^{d+1} (\grat{d, \rho} + 1)^{-d} = \rho,$$
the last equality coming for the definition of the generalized golden ratio $\grat{d,
\rho}$; thus, no player can unilaterally deviate from $\vecc s$ and gain (strictly) more than a factor of $\rho$.

Since $\vecc s$ is a $\rho$-equilibrium, the $\rho$-PoA of our game is at least
$$
\frac{C(\vecc s)}{\min_{\vecc s'}C(\vecc s')}
\geq \frac{C(\vecc s)}{C(\vecc s^*)}
=\frac{n \grat{d, \rho}^{d+1}}{n - 1+\frac{1}{\rho}\grat{d,
\rho}^{d+1}}
\to \grat{d, \rho}^{d+1},
$$
as $n$ grows arbitrarily large. This concludes our proof.
\end{proof}

\section{The Algorithm} \label{sec:algorithm}

In this section we describe and study our algorithm for computing $d^{d +
o(d)}$-approximate equilibria in weighted congestion games of degree $d$. The
algorithm, as well as the general outline of its analysis, are inspired by the work
of Caragiannis et al.~\cite{Caragiannis2015a}. However, as discussed in \cref{subsec:ourres}, here we
are using the approximate potential function of Christodoulou et al.~\cite{cggs2018}. This will
be crucial in proving that the algorithm is indeed poly-time and, more importantly,
has the improved approximation guarantee of $d^{d + o(d)}$.

In \cref{subseq:potent}, we introduce the aforementioned potential function
from~\cite{cggs2018}, along with some natural extensions that will be
useful for our analysis --- \emph{partial} and \emph{subgame} potentials. This
is complemented by a set of technical lemmas through which our use of this potential
function will be instantiated in the rest of the paper. In \cref{subseq:algorithm},
we describe our approximation algorithm and next, in \cref{subseq:runtime}, we show
that it indeed runs in polynomial time; the critical step in achieving this is
proving the ``Key Property'' (\cref{thm}) of our algorithm, which appropriately
bounds the potential of certain groups of players throughout its execution. Finally,
in \cref{subseq:approxFactor} we establish the desired approximation guarantee of
the algorithm.

\subsection{The Potential Function Technique} \label{subseq:potent}

Our subsequent proofs regarding both the runtime and the approximation guarantee of
our algorithm, will rely heavily on the use of an approximate potential function. In
particular, we will use a straightforward variant of the \emph{Faulhaber potential}
function, introduced by Christodoulou et al.~\cite{cggs2018}. Consider a polynomial weighted
congestion game of degree $d$. For every resource $e$ with cost function
$c_e(x)=\sum_{\nu=0}^da_{e,\nu}x^\nu$, we set
\begin{equation} 
\label{potentialDefine}
\pe{x} := a_{e, 0} \, x + \sum_{\nu=1}^{d}a_{e, \nu}
\left(x^{\nu+1}+\frac{\nu+1}{2} x^{\nu} \right).
\end{equation}
The potential of any state $\vecc{s}$ of the game is then defined as
$$
\p{}{\vecc{s}}{} := \sum_{e\in E} \pe{x_e(\vecc{s})}.
$$ 

The potential function introduced above satisfies a crucial property which is given
in the following lemma and which will be extensively used in the rest of our paper.
To state it, we will first define\footnote{The reason for introducing extra notation
here, and not making the arguably simpler choice to directly use the actual value of
$d+1$ instead of variable $\alpha$, for the rest of the paper, is that we want to
assist readability: by using $\alpha$, we make clear where exactly this factor from our
approximate potential comes into play.} an auxiliary constant
\begin{equation}
\label{eq:def_alphas}
\alpha := d + 1.
\end{equation}
\begin{lemma} 
\label{thelemma} 
For any resource $e$, $x \geq 0$ and $w \geq 1$:
\begin{equation*}
\label{eq:cond_approx_eq}
 w \cdot c_e(x+w) \leq \phi_e(x+w)-\phi_e(x) \leq \alpha \cdot w \cdot c_e(x+w),
\end{equation*} 
where function $\phi_e$ is defined given in \eqref{potentialDefine} and parameter $\alpha$ in~\eqref{eq:def_alphas}. 
\end{lemma}
\begin{proof}
From the proof\footnote{In particular, see Eq.~(4.10) in the proof of
\cite[Claim~4.7]{cggs2018} and the related \cite[Lemma~4.3]{cggs2018},
both applied for the special case of $\gamma=1$ here.} of Theorem~4.4
of Christodoulou et al.~\cite{cggs2018} we know that, if for resource $e$ with cost function
$c_e(x)=\sum_{\nu=0}^da_{e,\nu}x^\nu$ we define
\begin{equation}
\label{eq:def_pot_christo_etal}
\hat\phi_e(x)=\sum_{\nu=1}^d a_{e,\nu} S_\nu(x),
\qquad\text{where}\;\; 
S_\nu(x)=
\begin{cases}
\frac{x^{\nu+1}}{\nu+1}+\frac{x^\nu}{2},& \nu= 1,\dots,d,\\
x, &\nu=0,
\end{cases}
\end{equation}
then
$$
\frac{1}{d+1}\cdot w \cdot c_e(x+w) \leq \hat\phi_e(x+w) - \hat\phi_e(x)  \leq w \cdot c_e(x+w).
$$

We now multiply each $\nu$-th term of the sum for $\hat\phi_e(x)$
in~\eqref{eq:def_pot_christo_etal} by $\nu+1$, defining $$\phi_e(x)=\sum_{\nu=1}^d
a_{e,\nu} (\nu+1) S_\nu(x)$$ which coincides with our definition of the potential in
\eqref{potentialDefine}. It is not difficult to see (by following the same proof of
\cite[Thereom~4.4]{cggs2018}) that the above translates into essentially scaling the entire potential by $d+1$, resulting in the desired bound in the statement
of~\cref{thelemma}.
\end{proof}

The notions of \emph{partial} and \emph{subgame} potential were introduced in
\cite{Caragiannis2015a}, and we now adapt them for the setting of our new approximate
potential. The \emph{subgame potential} with respect to a group of players $R
\subseteq N$ of a state $\vecc{s}$ is defined as:
$$
\p{}{\vecc{s}}{R} := \sum_{e\in E} \pe{\xe{R}{\vecc{s}}}.
$$
The \emph{partial potential} with respect to a group of players $R \subseteq N$ is then defined as:
\begin{equation}
\label{eq:partial_pot_def}
\p{R}{\vecc{s}}{} := \p{}{\vecc{s}}{} - \p{}{\vecc{s}}{N \setminus R} = \sum_{e\in E} \left[\pe{x_e(\vecc{s})} -
\pe{\xe{N \setminus R}{\vecc{s}}} \right].
\end{equation}

In the original work by Christodoulou et al.~\cite{cggs2018}, the variable
$w$ in \cref{thelemma} was interpreted as a single player's weight, and $x$ as the
total remaining weight on resource $e$. In our analysis, however, $x$ and $w$ will
each play the role of the total weight of some group's players who are using
resource $e$. \Cref{thelemma} then becomes a very powerful algebraic tool that
relates social cost of groups of players and partial potentials in a variety of
settings, and plays an important role in many of our proofs. In many cases, it will
be present in such proofs implicitly, through the following corollary
(\cref{alphalemma}), which says that the chosen potential function is cost-revealing
in a strong sense. That is, for any group of players, the ratio of the partial
potential and the social cost of that group is bounded from both above and below.

\begin{lemma} 
\label{alphalemma} 
For any group of players $R \subseteq N$,
$$
\sum_{u \in R} \cu{\vecc{s}} \leq 
\p{R}{\vecc{s}}{} \leq \alpha \sum_{u \in R}\cu{\vecc{s}}.$$
\end{lemma}
\begin{proof}
First consider each resource $e$ separately. Setting $x = \xe{N \setminus
R}{\vecc{s}}$ and $w = \xe{R}{\vecc{s}}$ in \cref{thelemma}, we obtain
\begin{align*}
\xe{R}{\vecc{s}} \cdot \ce{\xe{N \setminus R}{\vecc{s}} + \xe{R}{\vecc{s}}} &\leq
\pe{\xe{N \setminus R}{\vecc{s}} + \xe{R}{\vecc{s}}} - \pe{\xe{N \setminus
R}{\vecc{s}}} \\ &\leq \alpha \cdot \xe{R}{\vecc{s}} \cdot \ce{\xe{N \setminus
R}{\vecc{s}} + \xe{R}{\vecc{s}}}.
\end{align*} Summing over all resources, we get that
\begin{align*}
\sum_{u \in R} \cu{\vecc{s}} \leq \p{}{\vecc{s}}{} - \p{}{\vecc{s}}{N \setminus R}
\leq \alpha \sum_{u \in R} \cu{\vecc{s}}.
\end{align*} By definition~\eqref{eq:partial_pot_def}, $\p{R}{\vecc{s}}{} = \p{}{\vecc{s}}{} - \p{}{\vecc{s}}{N
\setminus R}$, concluding the proof.
\end{proof}

The following lemma provides a relation between the change in potential due to a single
player's deviation and a linear combination of that player's old and new costs. For
an exact potential function, the latter would simply be the difference in the cost
experienced by that player. However, in our case, the player's cost in one
of the two states is weighted by an additional factor of $\alpha$ compared to her
cost in the other state. 
Thus, \cref{potentialProp} only implies a decrease in the potential function 
if the deviating player has improved her cost by at least a factor $\alpha$.

\begin{lemma} 
\label{potentialProp} 
Let $u$ be a player, $R$ be an arbitrary subset
of players with $u \in R$, and $\vecc{s}$ and $\vecc{s}'$ be two states that differ
\emph{only} in the strategy of player $u$. Then, 
$$
\p{R}{\vecc{s}}{} - \p{R}{\vecc{s}'}{} \geq 
\cu{\vecc{s}} - \alpha \cu{\vecc{s}'}.$$
\end{lemma}
\begin{proof}
First observe that, since each player in $N \setminus R$ plays the same strategy in both profiles $\vecc s$ and $\vecc s'$, it must be that  $\xe{N \setminus R}{\vecc{s}} = \xe{N \setminus R}{\vecc{s}'}$. Using this, we can derive that
$$
\p{}{\vecc{s}}{N \setminus R} 
= \sum_{e \in E} \pe{\xe{N \setminus R}{\vecc{s}}} 
= \sum_{e\in E} \pe{\xe{N \setminus R}{\vecc{s}'}} 
= \p{}{\vecc{s}'}{N \setminus R}, 
$$ 
and thus
$$
\p{R}{\vecc{s}}{} - \p{R}{\vecc{s}'}{} 
= \left[\p{}{\vecc{s}}{} - \p{}{\vecc{s}}{N
\setminus R}\right] - \left[\p{}{\vecc{s}'}{} - \p{}{\vecc{s}'}{N \setminus R}\right] 
=\p{}{\vecc{s}}{} - \p{}{\vecc{s}'}{},
$$
where the first equality is due to the definition of partial
potentials~\eqref{eq:partial_pot_def}. 

Finally, due to~\cref{thelemma}  (via Lemma~4.1 of Christodoulou et al.~\cite{cggs2018}) we know that $\varPhi$ is indeed an $\alpha$-approximate potential, that is, 
$$
\p{}{\vecc{s}}{}-\p{}{\vecc{s}'}{}  \geq
\cu{\vecc{s}}- \alpha \cu{\vecc{s}'}
$$
given that $\vecc s$ and $\vecc s'$ differ only on the strategy of player $u$.
\end{proof}

\subsection{Description of the Algorithm}
\label{subseq:algorithm}

We shall now describe our algorithm for finding $d^{d+o(d)}$-approximate equilibria
in weighted congestion games of degree $d$ (see \cref{main_alg}). We remark, once again, that it is inspired by a similar algorithm by Caragiannis et al.~\cite{Caragiannis2015a}. However, a critical difference is that our algorithm runs directly in the \emph{actual} game (using the original cost functions, and thus, players' deviations that are best-responses with respect to the actual game); as a result, we also need to appropriately calibrate the original parameters from Caragiannis et al.~\cite{Caragiannis2015a}. 
First, we fix the following constant that essentially captures our \emph{target approximation factor}:
\begin{equation}
\label{eq:pdef}
p := (2d+3)(d+1)(4d)^{d+1}=d^{d+o(d)}.
\end{equation}
The following lemma captures a critical property of the above parameter, which is the one that will essentially give rise to the specific approximation factor of our algorithm (see \cref{equil}).
\begin{lemma} \label{lemma_p_gamma}
For any positive integer $d$, the parameter $p$ defined in \eqref{eq:pdef} satisfies the following property:
$$
p \geq (2\alpha+1)\alpha\cdot\poa_d\left(\alpha+\frac{1}{p}\right)
=
(2\alpha+1)\alpha\Phi_{d,\alpha+\frac{1}{p}}^{d+1}
$$
where $\alpha$ is given in \eqref{eq:def_alphas}.
\end{lemma}
\begin{proof}
First, recall from~\eqref{eq:def_alphas} that $\alpha = d + 1$.
Next, we note that the Lambert-W function is increasing on the positive reals (see Corless et al.~\cite{Corless1996}) and so, for any $d\geq 1$, 
$$
\lambertW{\frac{d}{d + 2}} \geq \lambertW{\frac{1}{3}} \approx 0.258 > \frac{1}{4}.
$$
Furthermore, the $\rho$-approximate Price of Anarchy $\poa_d(\rho)$ is also nondecreasing with respect to the approximation parameter $\rho\geq 1$ (since the set of allowable approximate equilibria gets larger; see \cref{sec:notation}). Thus, from \cref{th:poa_upper} we can see that
$$
\poa_d\left(\alpha+\frac{1}{p}\right)
\leq 
\poa_d\left(d+2\right)
\leq \left[\frac{d}{\lambertW{\frac{d}{d + 2}}} \right]^{d+1}
\leq (4d)^{d+1}.
$$
From this, we can finally get that indeed
$$
(2d+3)(d+1)\poa_d\left(\alpha+\frac{1}{p}\right) \leq (2d+3)(d+1)(4d)^{d+1}=p.
$$
The equality in the statement of our lemma is just a consequence of the fact that the $\rho$-PoA is exactly equal to $\Phi_{d,\rho}^{d+1}$ (from \cref{sec:poa}; see \cref{th:poa_upper} and \cref{th:poa_lower}).
\end{proof}

\begin{algorithm}[t]
\caption{Computing $d^{d+o(d)}$-approximate equilibria in weighted polynomial congestion games of degree $d$}
\label{main_alg}
\begin{algorithmic}[1]
	\INPUT A polynomial weighted congestion game $\Gamma$ of degree $d$; A state $\vecc{s}_{\text{init}}$ of $\Gamma$
	\OUTPUT A $d^{d+o(d)}$-equilibrium state $\vecc{s}$
		\State $c_{\max} \gets \max_{u \in N}\cu{\vecc{s}_{\text{init}}}$, $c_{\min} \gets \min_{u \in N} \cu{\vecc{0}_{-u},\br{\vecc 0}}$, $m \gets \log \frac{c_{\max}}{c_{\min}}$
		\State $g \gets n p^{3} (1 + m(1 + p))^{d} d^{d}+1$\label{algo_g_def} 
		\Comment{$p$ is defined in \eqref{eq:pdef}}
		\State $b_i \gets g^{-i} c_{\max}$ \quad for all $i=1,\dots,m$
		\State $\vecc{s} \gets \vecc{s}_{\text{init}}$
		\State $N_{\text{fixed}} \gets \emptyset$
		\While {$\exists\, u \in N: \; \cu{\vecc{s}} \geq b_1 \;\land\; u\; \text{has an $(\alpha + \frac{1}{p})$-move}$}
			\State $\vecc{s} \gets (\vecc{s}_{-u}, \br{\vecc{s}})$
		\EndWhile
		\For {$i = 1$ to $m - 1$}
			\While { %
			 $\exists u \in N\setminus N_{\text{fixed}}:$ \par 
			 \hskip\algorithmicindent $\Big[\cu{\vecc{s}}\in [b_{i+1}, b_i) \;\land\; u\;\text{has an $(\alpha + \frac{1}{p})$-move}\Big] \biglor
			 \Big[\cu{\vecc{s}} \geq b_i\;\land\; u\;\text{has a $p$-move}\Big]$} \label{algo_improvements}
			\State $\vecc{s} \gets (\vecc{s}_{-u}, \br{\vecc{s}})$ \label{algo_improvements_gets}
		\EndWhile
		\State $N_{\text{fixed}}\gets N_{\text{fixed}}\union\sset{u \in N \setminus N_{\text{fixed}} : \cu{\vecc{s}} \geq b_i}$ \label{algo_nfixed}
		\EndFor
		\smallskip
		\State $N_{\text{fixed}}\gets N_{\text{fixed}}\union\sset{u \in N \setminus N_{\text{fixed}} : \cu{\vecc{s}} \geq b_m}$
\end{algorithmic}
\end{algorithm}
The input to the algorithm consists of the description of a weighted polynomial congestion game
$\Gamma$ of fixed degree $d$ and an (arbitrary) initial state $\vecc{s}_{\text{init}}$ of $\Gamma$. 
We now let $c_{\max} := \max_{u \in N}
\cu{\vecc{s}_{\text{init}}}$ and $c_{\min} := \min_{u \in N} \cu{\vecc{0}_{-u},
\br{\vecc{0}}}$. Here, $\vecc{0}$ denotes the ``empty state'', that is, a
fictitious state in which all players' strategies are empty sets. Recall also
(see~\cref{sec:notation}) that $\br{\cdot}$ returns a best-response move of player
$u$ at a given state. Observe that $c_{\max}$ is the maximum cost of a player in
state $\vecc{s}_{\text{init}}$, and $c_{\min}$ can be used as a lower bound on the
cost of any player in any state of our game. We also define parameter $m =
\log \frac{c_{\max}}{c_{\min}}$. Notice that $m$ is polynomial on the input (that
is, the description of our game $\Gamma$). Next, we introduce a factor $g$ (see
\cref{algo_g_def} of \cref{main_alg}), that depends (polynomially) on both the
aforementioned parameter $m$ of our game and the approximation factor $p$. Using
this, we set $m+1$ ``boundaries'' $c_{\max}=b_0,b_1, \ldots, b_m$ so that $b_i =
g^{-i} c_{\max}$ for the player costs (see discussion below).

The algorithm runs in $m$ \emph{phases}, indexed by $i=0,1,\dots,m-1$. It is
helpful to introduce here the following notation. We denote by $\vecc{s}^{i}$ the
state of the game immediately after the end of phase $i$. Furthermore, let $R_i$ be
the set of players who were at least once selected by the algorithm to make an
improvement move during phase $i$.

Phase $i=0$ is ``preparatory'': we repeatedly select a player whose cost exceeds the
largest (non-trivial) boundary $b_1$ and who has an $(\alpha + p^{-1})$-move, and
allow her to make a best-response move. This phase ends when there are no such
players left.

Next follow the ``main'' $m-1$ phases.  The algorithm constructs the final state
$\vecc s^{m-1}$, which is the sought $d^{d+o(d)}$-approximate equilibrium, step by
step. That is, it finalizes players' strategies in ``packets'' rather than
one-by-one. In particular, during each phase it fixes the strategies of a certain
group of players, and never changes them again.

The $i$-th phase itself, consists of a sequence of best-responses. We allow players
(who are not ``fixed'' yet) to repeatedly make best-response moves, according to a
certain rule: if a player's current cost is ``large'', i.e.\ lies in  $[b_i,\infty)$, she
is allowed to best-respond only if she can improve at least by a factor of $p$ (that is, she
has a $p$-move); if, however, her cost is ``small'', i.e.\ in  $[b_{i+1},b_i)$, she
can play even if she only has an $(\alpha + p^{-1})$-move.\footnote{One can think of
this second type of moves as ``preparatory'' for the next phase. Indeed, they ensure
that immediately before the beginning of the $(i+1)$-th phase, players who are not
yet fixed and whose costs exceed $b_{i+1}$ are in an $(\alpha +
p^{-1})$-equilibrium. We can expect such players to perform relatively few $p$-moves
during the $(i+1)$-th phase, since the approximation factor $\alpha + p^{-1}$ is
significantly smaller than $p$ (see \eqref{eq:pdef}).} The phase ends when there are
no players left to make a deviation according to the above rule.
At this point, the algorithm fixes the strategies of
all players whose \emph{current} cost lies in $[b_i,\infty)$, and adds them to the
set $N_{\text{fixed}}$.

Immediately after the end of phase $m-1$, the algorithm returns the current state of
the game and terminates. Observe that after the final phase of the algorithm, all
players have been included in $N_{\text{fixed}}$: using the definitions of $m$ and
$b_m$, it is not difficult to see that $b_m \leq c_{\min}$.

\subsection{Running Time} \label{subseq:runtime}
In this section, we prove the ``Key Property'' of our algorithm (\cref{thm}), and
then establish \cref{runtime} which guarantees that the runtime of the algorithm is
polynomial. We remark here that we assume, for the runtime analysis, that
best-response strategies are efficiently computable.\footnote{Formally, we want
function $\br{\vecc s_{-u}}$ to be computable in polynomial time, for any player $u$ and
any strategies $\vecc s_{-u}$ of the other players. This assumption is necessary in
general, since the number of strategies of a player can, in principle, be
exponential in the input size. Think, for example, of source-sink paths in network
congestion games. Of course, if the strategy sets of the players have polynomial
size in the first place, then the time-efficiency of best-responses comes for free
anyway.} One should also keep in mind that the degree $d$ of the game is considered
a \emph{constant}. However, with respect to the analysis of the approximation
guarantee of the algorithm (see the following~\cref{subseq:approxFactor}), it will
actually be treated as a parameter: the approximation factor will be given as a
function of $d$. We want to emphasize here that all these assumptions are standard
in the literature of algorithms for approximate equilibria in congestion games (see
e.g. \cite{Caragiannis2011,Caragiannis2015a,Feldotto2014,Feldotto2017}).

Recall that $R_i$ is the set of players that made at least one improvement move
during the $i$-th phase of the algorithm. In the following it will be useful to
denote by $C_u(i)$, for any $u\in R_i$, the cost of player $u$ immediately after her
last move within phase $i$. Then we can prove the following upper bound on the
partial potential of $R_i$ immediately after phase $i$:

\begin{lemma} \label{lemmaKeyProperty} 
For every phase $i$, $\p{R_i}{\vecc{s}^i}{} \leq
\alpha \sum_{u \in R_i} C_u(i)$.
\end{lemma}
\begin{proof}
Rename players in $R_i$ as $u_1, ..., u_{|R_i|}$ 
in increasing order of time of their last move 
during phase $i$. Call $\vecc{s}^{i,j}$ the state immediately after player $u_j$ made her
last move during phase $i$. 
Define $R_i^j := \{u_{j+1}, ..., u_{|R_i|}\}$, the set of players who moved during phase $i$ \emph{after}  state $\vecc{s}^{i,j}$. 
Then, using the definition of the partial potential \eqref{eq:partial_pot_def}, we can obtain the telescoping sum 
\begin{equation}
\label{eq:telescope_helper}
\p{R_i}{\vecc{s}^i}{} = \p{}{\vecc{s}^i}{} - \p{}{\vecc{s}^i}{N \setminus R_i}
= \sum_{j=1}^{|R_i|} \left[\p{}{\vecc{s}^i}{N \setminus R_i^j} 
- \p{}{\vecc{s}^i}{N \setminus R_i^{j-1}} \right] 
\end{equation}

We now consider each term of the above sum separately. For any $j =1, \dots, |R_i|$
and any resource $e$, we have $\xe{N \setminus R_i^j}{\vecc{s}^{i,j}} = \xe{N \setminus
R_i^j}{\vecc{s}^{i}}$ and $\xe{N \setminus R_i^{j-1}}{\vecc{s}^{i,j}} = \xe{N \setminus
R_i^{j-1}}{\vecc{s}^{i}}$. Indeed, only players in $R_i^j$ made moves between state
$\vecc{s}^{i,j}$ and the end of phase $i$, and consequently, players in $N \setminus R_i^j$
and in $N \setminus R_i^{j-1} \subseteq N \setminus R_i^j$ never changed their
strategies between states $\vecc{s}^{i,j}$ and $\vecc{s}^{i}$. Thus,
\begin{align*}
\p{}{\vecc{s}^i}{N \setminus R_i^j} 
- \p{}{\vecc{s}^i}{N \setminus R_i^{j-1}} 
& = \sum_{e\in E} \left[\pe{\xe{N \setminus R_i^j}{\vecc{s}^{i}}} 
- \pe{\xe{N \setminus R_i^{j - 1}}{\vecc{s}^{i}}}\right]\\
&= \sum_{e\in E} \left[\pe{\xe{N \setminus R_i^j}{\vecc{s}^{i,j}}} 
- \pe{\xe{N \setminus R_i^{j - 1}}{\vecc{s}^{i,j}}}\right]
\end{align*}
At this point, we use that $N \setminus R_i^j = (N \setminus R_i^{j-1}) \union \{u_j\}$, paired with \cref{thelemma} and the monotonicity of $c_e$, to further obtain
\begin{align*}
\p{}{\vecc{s}^i}{N \setminus R_i^j} 
- \p{}{\vecc{s}^i}{N \setminus R_i^{j-1}} 
& = \sum_{e\in E} \left[\pe{\xe{N \setminus R_i^{j-1}}{\vecc{s}^{i,j}} + \xe{u_j}{\vecc{s}^{i,j}}} 
- \pe{\xe{N \setminus R_i^{j-1}}{\vecc{s}^{i,j}}}\right] \\ 
& \leq \sum_{e\in E} \alpha \xe{u_j}{\vecc{s}^{i,j}} \cdot 
\ce{\xe{N \setminus R_i^{j-1}}{\vecc{s}^{i,j}} 
+ \xe{u_j}{\vecc{s}^{i,j}}} \\ 
& \leq \alpha\sum_{e\in E}  \xe{u_j}{\vecc{s}^{i,j}} \cdot \ce{x_e(\vecc{s}^{i,j})} \\
& = \alpha  C_{u_j}(\vecc{s}^{i,j}).
\end{align*}
Applying this bound to each term in the sum in \eqref{eq:telescope_helper}, we finally get the desired inequality
\begin{align*}
\p{R_i}{\vecc{s}^i}{} 
\leq \, \sum_{j=1}^{|R_i|} \alpha \, C_{u_j}(\vecc{s}^{i,j})
= \alpha \sum_{u \in R_i} C_u(i),
\end{align*}
where for the last equality we used the definition of $C_u(i)$.
\end{proof}

The following result is necessary to prove \cref{thm} and is essentially the
cornerstone of the approach to computing approximate equilibria that our algorithm
takes. Namely, \cref{rhoStretch} considers a group of players in an approximate
equilibrium at some state $\vecc{s}$ and shows that the potential of $\vecc{s}$ is
at most a fixed factor away from the potential of any other state that only differs
from $\vecc{s}$ in the strategies of players in that group. In particular, this
factor turns out to be of the order of the $\rho$-PoA (see
also~\cref{th:poa_upper}).

\begin{lemma} \label{rhoStretch} 
For any group of players $R \subseteq N$, consider any two states $\vecc{s}$ and
$\vecc{s}'$ such that $\vecc{s}$ is a $\rho$-equilibrium state for the players in
$R$ and every player in $N \setminus R$ uses the same strategy in $\vecc{s}$ and
$\vecc{s}'$. Then $\p{R}{\vecc{s}}{} \leq \alpha  \grat{d,\rho}^{d+1} 
\p{R}{\vecc{s}'}{}$.
\end{lemma}
\begin{proof}
Utilizing~\cref{alphalemma} we can bound the ratio of the partial potentials by 
$$\frac{\p{R}{\vecc{s}}{}}{\p{R}{\vecc{s}'}{}} 
\leq \frac{\alpha \sum_{u \in
R} \cu{\vecc{s}}}{\sum_{u \in R} \cu{\vecc{s}'}}  
\leq \alpha \, \grat{d,\rho}^{d+1},$$ where the
last inequality holds from by \cref{partialpoa}.
\end{proof}

Finally, the following theorem establishes that the partial potential of the players
$R_i$, who are going to move during the $i$-th phase, is linearly-bounded by the
$i$-th cost boundary $b_i$ at the beginning of that phase. This result will be used in two ways: on one hand, it is the
basis of the proof of \cref{runtime}, which guarantees polynomial runtime of the
algorithm; on the other hand, \cref{thm} will be used in \cref{subseq:approxFactor}
as well, to bound the approximation factor of the algorithm (\cref{relate_fix_end}
in particular).
\begin{theorem}[Key Property] \label{thm} 
For every phase $i \geq 1$ of the algorithm,
$\p{R_i}{\vecc{s}^{i-1}}{} \leq npb_i.$
\end{theorem}
\begin{proof}
We split $R_i = P_i \cup Q_i$, where $P_i$ are players in $R_i$ whose last move
during phase $i$ was a $p$-move, and $Q_i$ those whose last move during phase $i$
was an $(\alpha + p^{-1})$-move. Now suppose there were $\tau$ best-response moves
executed during the course of phase $i$ (see \cref{algo_improvements_gets} in
\cref{main_alg}), and denote the corresponding states in the game by $\vecc{s}^{i-1}
= \vecc{t}_0, \vecc{t}_1, \ldots, \vecc{t}_{\tau} = \vecc{s}^i$. Then
\begin{equation} \label{eq:sumTelescoping}
\p{R_i}{\vecc{s}^{i-1}}{} - \p{R_i}{\vecc{s}^{i}}{} 
= \sum_{j=1}^{\tau} \left[ \p{R_i}{\vecc{t}_{j-1}}{} - \p{R_i}{\vecc{t}_j}{}\right].
\end{equation}
Notice here that, due to~\cref{potentialProp} and the fact that $p, \alpha+p^{-1}>
\alpha$, i.e., all best-responses in the above sum are better than $\alpha$-moves,
all terms of the above sum are positive. Thus, keeping only those terms $j$ for
which $\vecc{t}_{j-1} \rightarrow \vecc{t}_j$ corresponds to the last move during
phase $i$ of some player in $P_i$, we get that
\begin{equation} \label{ineqThruSumP}
\p{R_i}{\vecc{s}^{i-1}}{} - \p{R_i}{\vecc{s}^{i}}{} \geq \sum_{u \in P_i} (p-\alpha) C_u(i),
\end{equation}
where recall that $C_u(i)$ denotes the cost of player $u$ after her last move within phase $i$.

By \cref{lemmaKeyProperty}, we also have that
$$
\sum_{u \in P_i} C_u(i) 
= \sum_{u \in R_i} C_u(i) - \sum_{u \in Q_i} C_u(i) 
\geq \frac{1}{\alpha} \p{R_i}{\vecc{s}^i}{} 
- \sum_{u \in Q_i} C_u(i) 
\geq \frac{1}{\alpha} \p{R_i}{\vecc{s}^i}{} 
- n b_i, 
$$
the last inequality holding due to the fact that all players in $Q_i$ have cost at
most $b_i$ at the end of phase $i$. Combining this with \eqref{ineqThruSumP}, and
solving with respect to $\p{R_i}{\vecc{s}^{i}}{} $, we finally get that
$$
 \p{R_i}{\vecc{s}^{i}}{} 
\leq \frac{\alpha}{p}\left[\p{R_i}{\vecc{s}^{i-1}}{} + (p-\alpha) nb_i\right]
\leq \frac{\alpha}{p}\left[\p{R_i}{\vecc{s}^{i-1}}{} + pnb_i\right].
$$
To arrive at a contradiction, and complete the proof of our theorem, from now on assume that $\p{R_i}{\vecc{s}^{i-1}}{} > nb_ip$. Then, the above inequality becomes
\begin{equation} \label{ineqAssumption}
\p{R_i}{\vecc{s}^{i}}{} 
< \frac{2\alpha}{p}
\p{R_i}{\vecc{s}^{i-1}}{}.
\end{equation}

We now introduce some additional notation for the current proof. First, partition
$R_i$ into $X_i$ and $Y_i$, where $X_i$ contains the players in $R_i$ whose cost is
at least $b_i$ at state $\vecc{s}^{i-1}$, and $Y_i$ those whose cost is strictly
less than $b_i$. Secondly, consider a helper game\footnote{This is exactly the same
construction as in the proof of Caragiannis et al.~\cite[Lemma~4.3]{Caragiannis2015a}.} $\bar \Gamma$
with the same set of players $\bar N=N$ and resources $\bar E=E \cup \{e_u: u \in
Y_i\}$. The new, extra resources have costs
$$
\bar c_{e_u}(t) = \frac{b_i}{w_u}\qquad \text{for all}\;\; u \in Y_i,
$$ 
while the old ones have the same costs $\bar c_{e}(t) = c_e(t)$ for all $e\in E$.
All players in $N\setminus Y_i$, have the same strategies in $\bar\Gamma$ as before
in $\Gamma$. However, for the rest we set
$$
\bar S_u=\ssets{s_u}\union \sset{s_u'\union\ssets{e_u}\fwh{s_u'\in
S_u\setminus\ssets{s_u}}}\qquad \text{for all}\;\; u \in Y_i,
$$
where $s_u$ is the strategy that $u$ uses in state $\vecc{s}^{i-1}$ of the original
game $\Gamma$. Denote the potential in $\Gamma'$ by $\bar{\varPhi}$.

We define two special strategy profiles, namely $\bar{\vecc{s}}^{i-1}$ and
$\bar{\vecc{s}}^{i}$, of our new game $\bar\Gamma$.
In state $\bar{\vecc{s}}^{i-1}$, every player uses the same strategy as in state
$\vecc{s}^{i-1}$ of the original game $\Gamma$.
In state $\bar{\vecc{s}}^{i}$, again players use the same strategies as in state
$\vecc{s}^{i}$ of $\Gamma$, except from those players in $Y_i$ that happen to have
different strategies between $\vecc{s}^{i-1}$ and $\vecc{s}^{i}$. In particular, if
a player $u\in Y_i$ uses strategy $s_u$ under profile $\vecc{s}^{i-1}$ and strategy
$s'_u\neq s_u$ at profile $\vecc{s}^{i}$, then we set her strategy at
$\bar{\vecc{s}}^{i}$ to be $s'_u \cup \{e_u\}$.

Then, one can show (see~\cref{app:proof_key}) the following about these states:
first, $\bar{\vecc{s}}^{i-1}$ is an $(\alpha + p^{-1})$-equilibrium (of game
$\bar\Gamma$) for all players in $R_i$; and secondly,
$$
\pbar{R_i}{\bar{\vecc{s}}^{i}}{}
\leq \p{R_i}{\vecc{s}^{i}}{} + nb_i.
$$
Using \eqref{ineqAssumption} and recalling our assumption that
$\p{R_i}{\vecc{s}^{i-1}}{} > nb_ip$, the above inequality gives us that
\begin{equation*}
\pbar{R_i}{\bar{\vecc{s}}^{i}}{} < \frac{2\alpha + 1}{p}
\p{R_i}{\vecc{s}^{i-1}}{}
=
\frac{2\alpha + 1}{p}
\pbar{R_i}{\bar{\vecc{s}}^{i-1}}{}, 
\end{equation*} 
where the last equality is a result of the fact that all players use
the same strategies in $\bar{\vecc{s}}^{i-1}$ and $\vecc{s}^{i-1}$, which consist
only of edges that have exactly the same cost functions at $\Gamma$ and
$\bar\Gamma$.
From~\cref{lemma_p_gamma}, this becomes
$$
\pbar{R_i}{\bar{\vecc{s}}^{i-1}}{} 
> \alpha \Phi_{d,\alpha+\frac{1}{p}}^{d+1} \pbar{R_i}{\bar{\vecc{s}}^{i}}{},
$$ 
which contradicts~\cref{rhoStretch}, since all players in $R_i$ are in an $(\alpha +
p^{-1})$-equilibrium in $\bar{\vecc{s}}^{i-1}$, and every player in $N \setminus
R_i$ has the same strategy in $\bar{\vecc{s}}^{i-1}$ and $\bar{\vecc{s}}^{i}$.
This contradiction concludes our proof.
\end{proof}

Now we are able to prove the main theorem of this section:
\begin{theorem} \label{runtime} 
\cref{main_alg} runs in time polynomial in the number
of bits of the representation of its input game $\Gamma$.
\end{theorem}

\begin{proof}
It suffices to show that every phase $i=0, \dots, m-1$ runs in
poly-time, since the number of phases $m = \log \frac{c_{\max}}{c_{\min}}$ is
polynomial in the representation of the game.

First, fix such a phase $i$ and suppose there were $\tau$ single-player
best-response moves (see \cref{algo_improvements_gets} of \cref{main_alg}) during
the course of this phase. For this proof only, set $\vecc{s}^{-1} :=
\vecc{s}_{\text{init}}$. Denote all intermediate states that correspond to these
deviations by $\vecc{s}^{i-1} =
\vecc{t}_0, \vecc{t}_1,
\ldots, \vecc{t}_{\tau} = \vecc{s}^i$ in chronological order, and call $u_j$ the
player who makes the best-response move corresponding to the change from state
$\vecc{t}_{j-1}$ to $\vecc{t}_{j}$. Then
\begin{equation} \label{telscope}
\p{}{\vecc{s}^{i-1}}{} - \p{}{\vecc{s}^{i}}{} = \p{R_i}{\vecc{s}^{i-1}}{} -
\p{R_i}{\vecc{s}^{i}}{} = \sum_{j=1}^{\tau} \left[ \p{R_i}{\vecc{t}_{j-1}}{} -
\p{R_i}{\vecc{t}_j}{} \right].
\end{equation}
The first equality holds using the definition of the partial
potential~\eqref{eq:partial_pot_def} and the fact that players in $N \setminus R_i$
do not move during phase $i$, so $\p{}{\vecc{s}^{i-1}}{N \setminus R_i} =
\p{}{\vecc{s}^{i}}{N \setminus R_i}$. The second uses a telescoping sum.

Now observe that during the phase, players selected to make a move either make an
$(\alpha + p^{-1})$-move or a $p$-move (for the $i=0$, only the former). Hence, for
any intermediate state $j$, it must be that
$$
\cuj{\vecc{t}_{j - 1}} \geq \min\ssets{p,a+p^{-1}}\cdot \cuj{\vecc{t}_{j}} =
(a+p^{-1}) \cuj{\vecc{t}_{j}}
$$
Applying \cref{potentialProp}, we then obtain
\begin{equation} \label{eq:deviation}
\p{R_i}{\vecc{t}_{j-1}}{} - \p{R_i}{\vecc{t}_j}{}
\geq \cuj{\vecc{t}_{j-1}} - \frac{\alpha}{\alpha+p^{-1}} \cuj{\vecc{t}_{j-1}} =
\frac{1}{\alpha p + 1} \cuj{\vecc{t}_{j-1}}.
\end{equation}
By definition of the algorithm (see \cref{main_alg}), player $u_j$ has to have cost
$\cuj{\vecc{t}_{j-1}} \geq b_{i+1}$, otherwise she wouldn't have made any moves during
phase $i$. Recalling that  $b_{i+1} = b_{i}g^{-1}$, inequality \eqref{eq:deviation}
now tells us that, for all $j=1,\dots,\tau$,
\begin{equation*}
\p{R_i}{\vecc{t}_{j-1}}{} - \p{R_i}{\vecc{t}_j}{}
\geq \frac{b_{i} g^{-1}}{\alpha p + 1}
\end{equation*} and so applying \eqref{telscope} we get that
$$
\p{R_i}{\vecc{s}^{i-1}}{} - \p{R_i}{\vecc{s}^{i}}{} \geq \tau \frac{b_{i} g^{-1}}{\alpha p
+ 1}.
$$
If $i\geq 1$, this can be rewritten as
$$
\tau \leq g(\alpha p +1) b_i^{-1}\cdot \left[\p{R_i}{\vecc{s}^{i-1}}{} -
\p{R_i}{\vecc{s}^{i}}{}\right]
\leq g(\alpha p +1) \cdot b_i^{-1}\p{R_i}{\vecc{s}^{i-1}}{}
\leq n\cdot g(\alpha p +1)p,
$$
where the last inequality holds due to~\cref{thm}.

For the remaining, special case of $i=0$, since $b_0=c_{\max}$ and
$\vecc{s}^{-1}={\vecc s}_{\text{init}}$, we similarly get that
$$
\tau \leq g(\alpha p +1) \cdot c_{\max}^{-1}\p{}{\vecc{s}_{\text{init}}}{}
\leq n\cdot \alpha g(\alpha p +1),
$$
where the last inequality is a consequence of applying~\cref{alphalemma} (with
$\vecc s=\vecc{s}_{\text{init}}$ and $R=N$):
$$
\p{}{\vecc{s}_{\text{init}}}{}
\leq \alpha \sum_{u \in N} \cu{\vecc{s}_{init}}
\leq \alpha n c_{\max}.
$$

In any case, we showed that indeed $\tau$ is polynomially bounded in the size of the
input.
\end{proof}

\subsection{Approximation Factor} \label{subseq:approxFactor}

In this section, we establish the desired approximation guarantee of our algorithm
(\cref{equil}). Our high-level approach is to examine the amount by which the cost of a player
can change after she is ``fixed'' by the algorithm, and in particular, how the
changes in other players' strategies during the following phases affect the beneficial
deviating moves that she might have at the final state. We formalize this in
\cref{relate_fix_end}, which is the backbone of proving \cref{equil}, together with
the Key Property (\cref{lemmaKeyProperty}) of the previous section.

The following theorem establishes two very important properties. First, after
the strategy of a player $u$ becomes fixed, immediately after some phase $j$, and
until the end of the algorithm (phase $\vecc{s}^{m-1}$), her cost will not increase
by more than a small factor; and furthermore,
no deviation $s'_u$ of player $u$ can be much more profitable at the end of the
algorithm than it was immediately after the fixing phase $j$:

\begin{theorem} \label{relate_fix_end}
Suppose the strategy of player $u$ was fixed at the end of phase $j$. Then,
\begin{equation}
\label{eq:relate_fix_end_1}
\cu{\vecc{s}^{m-1}} \leq (1+3p^{-1})  \cu{\vecc{s}^{j}}
\end{equation}
and
\begin{equation}
\label{eq:relate_fix_end_2}
\cu{\vecc{s}_{-u}^{m-1}, s'_u} \geq (1-2p^{-1}) \cu{\vecc{s}_{-u}^{j}, s'_u}
\qquad\text{for all}\;\; s_u'\in S_u.
\end{equation}
\end{theorem}
To prove our theorem, we will need to capture the relation of the cost of player $u$
between state $\vecc s^j$ (when she was fixed) and $\vecc s^{m-1}$ (end state of the
algorithm), in a more refined way. In particular, we want to see how the cost of
player $u$ can change during all intermediate phases $i=j+1,\dots,m-1$. This is
exactly captured by the following lemma, proven in~\cref{app:proof_lemma2}:
\begin{lemma} \label{lemma2}
Fix any $\varepsilon > 0$ and let
$\xi_{\varepsilon} 
:= \left(1+\frac{1}{\varepsilon} \right)^d d^d$. 
Consider any player $u$ whose strategy was fixed immediately after phase $j$. 
Then, for any phase $i > j$, 
\begin{equation}
\label{eq:lemma2_1}
\cu{\vecc{s}^{i}} \leq (1+\varepsilon) \cu{\vecc{s}^{i-1}} 
+ \xi_{\varepsilon} \p{R_i}{\vecc{s}^{i-1}}{}
\end{equation}
and
\begin{equation}
\label{eq:lemma2_2}
\cu{\vecc{s}_{-u}^{i}, s'_u} 
\geq \frac{1}{1+\varepsilon} \cu{\vecc{s}_{-u}^{i-1}, s'_u} 
- \frac{\xi_{\varepsilon}}{1 + \varepsilon} 
\p{R_i}{\vecc{s}^{i-1}}{}\qquad\text{for all}\;\; s'_u\in S_u.
\end{equation} 
\end{lemma}
\Cref{lemma2} states that the cost of player $u$ cannot increase by too much during
any phase $i > j$ and, at the same time,  the cost player $u$ would get by deviating
cannot decrease by too much during any phase $i > j$.
We remark here that the upper bound provided by \cref{lemma2} will itself be bounded
by the guarantee provided by \cref{thm}.
Now, by appropriately applying \Cref{lemma2} in an iterative way, we can prove our
theorem:
\begin{proof}[Proof of~\cref{relate_fix_end}]
First, fix an $\epsilon >0$ such that 
\begin{equation}
\label{eq:epsilon_helper}
(1 + \varepsilon)^{m} = 1 + \frac{1}{p}.
\end{equation}
Then, by \cref{cl1} it must be that $m(1 + p) \geq \varepsilon^{-1}$, and so for the
parameter $g$ of our algorithm (see \cref{algo_g_def} of \cref{main_alg}) we have the bound
\begin{equation}
\label{eq:bound_g}
g-1 = n p^{3} (1 + m(1 + p))^{d} d^{d} \geq n p^{3} \xi_{\varepsilon},
\end{equation} 
where $\xi_\varepsilon$ is defined as in~\cref{lemma2}. 

Also, observe that, since $j$ is the phase at which player $u$ is fixed, it must be
that
\begin{equation}
\label{eq:props_player_fixed}
\cu{\vecc{s}^{j}} \geq b_j
\qquad\text{and}\qquad
\cu{\vecc{s}^{j}} \leq p \cdot \cu{\vecc{s}_{-u}^{j}, s'_u},
\end{equation}
for any deviation $s_u'\in S_u$ of player $u$. These are both consequences of the
definition of our algorithm: the first inequality comes immediately from
\cref{algo_nfixed} of \cref{main_alg}, while the second one is due to the fact that
player $u$ has no $p$-move left at the end of phase $j$ (see
\cref{algo_improvements} of \cref{main_alg}).

Next, for proving~\eqref{eq:relate_fix_end_1}, apply \eqref{eq:lemma2_1}
recursively, backwards for $i=m-1,\dots,j+1$, to get that
\begin{align}
\cu{\vecc{s}^{m-1}} 
& \leq  (1+\varepsilon)^{m-1-j} \cu{\vecc{s}^{j}} 
+ \xi_{\varepsilon} \sum_{i = j+1}^{m-1} 
(1+\varepsilon)^{m-1-i} \p{R_{i}}{\vecc{s}^{i-1}}{} \notag\\
& \leq  (1+\varepsilon)^{m} \left[\cu{\vecc{s}^{j}} 
+ \xi_{\varepsilon} \sum_{i = j+1}^{m-1} 
\p{R_{i}}{\vecc{s}^{i-1}}{} \right]\notag\\
& =  (1+p^{-1}) \left[\cu{\vecc{s}^{j}} 
+ \xi_{\varepsilon} \sum_{i = j+1}^{m-1} 
\p{R_{i}}{\vecc{s}^{i-1}}{} \right] \label{eq:bound_cost_last_1}, 
\end{align}
where for the second inequality we have used the fact that $m-1-j, m-1-i\leq m$, and
for the last equality we used~\eqref{eq:epsilon_helper}.
By using \cref{thm}, we can now bound the second term
of~\eqref{eq:bound_cost_last_1} by
\begin{equation}
\label{eq:bound_2nd_term}
\xi_{\varepsilon} \sum_{i = j+1}^{m-1} \p{R_{i}}{\vecc{s}^{i-1}}{} 
\leq \xi_{\varepsilon} np\sum_{i = j + 1}^{m-1} b_i
= \xi_{\varepsilon} np b_j \sum_{i = 1}^{m-1 - j} g^{-i}
\leq \frac{\xi_{\varepsilon} np}{g-1} b_j
\leq p^{-2} b_j
\leq p^{-2} \cu{\vecc{s}^{j}}.
\end{equation}
where the last two inequalities are due to~\eqref{eq:bound_g}
and~\eqref{eq:props_player_fixed}, respectively.

Combining \eqref{eq:bound_cost_last_1} with \eqref{eq:bound_2nd_term}, we finally
get
\begin{equation*}
\cu{\vecc{s}^{m-1}} \leq 
(1 + p^{-1})(1+p^{-2}) \cu{\vecc{s}^{j}}
=(1+p^{-1}+p^{-2}+p^{-3}) \cu{\vecc{s}^{j}}
,
\end{equation*}
which establishes \eqref{eq:relate_fix_end_1}.

Moving to \eqref{eq:relate_fix_end_2} now, and applying \eqref{eq:lemma2_2} recursively, 
backwards for $i=m-1,\dots,j+1$, we get that
\begin{align} \label{twoTerms}
\cu{\vecc{s}_{-u}^{m-1}, s'_u} &\geq  (1+\varepsilon)^{j+1-m} \cu{\vecc{s}_{-u}^{j},
s'_u}
- \xi_{\varepsilon} \sum_{i = j+1}^{m-1} (1+\varepsilon)^{i-m}
  \p{R_{i}}{\vecc{s}^{i-1}}{}\notag\\ &\geq
  (1+\varepsilon)^{-m}\cu{\vecc{s}_{-u}^{j}, s'_u} - \xi_{\varepsilon} \sum_{i =
  j+1}^{m-1} \p{R_{i}}{\vecc{s}^{i-1}}{}\notag\\ &=
  (1+p^{-1})^{-1}\cu{\vecc{s}_{-u}^{j}, s'_u} - \xi_{\varepsilon} \sum_{i =
  j+1}^{m-1} \p{R_{i}}{\vecc{s}^{i-1}}{},
\end{align} the second inequality coming from observing that $j+1-m\geq m$, $i-m\leq
0$ and the last one by using~\eqref{eq:epsilon_helper}.
Using~\eqref{eq:bound_2nd_term} and~\eqref{eq:props_player_fixed}, we can bound the
second term in \eqref{twoTerms} by
$$
\xi_{\varepsilon} \sum_{i = j+1}^{m-1} \p{R_{i}}{\vecc{s}^{i-1}}{}
\leq p^{-2} \cu{\vecc{s}^{j}}
\leq p^{-1} \cu{\vecc{s}_{-u}^{j}, s'_u},
$$
and thus \eqref{twoTerms} finally becomes
$$
\cu{\vecc{s}_{-u}^{m-1}, s'_u} 
\geq \left[(1+p^{-1})^{-1}-p^{-1}\right] \cu{\vecc{s}_{-u}^{j}, s'_u}
\geq (1-2p^{-1})\cu{\vecc{s}_{-u}^{j}, s'_u},
$$
where for the last inequality we used the fact that, for any real $x>0$,
$(1+x)^{-1}\geq 1-x$. This establishes \eqref{eq:relate_fix_end_2}, concluding the
proof of the theorem.
\end{proof}

We are now ready to prove the main theorem of this section:
\begin{theorem} \label{equil}
At the end of \cref{main_alg}, all players are in a $d^{d + o(d)}$-equilibrium.
\end{theorem}
\begin{proof}
Consider an arbitrary player $u$ and any strategy $s'_u$ she can deviate to from
state $\vecc{s}^{m-1}$. Suppose her strategy was fixed right after state
$\vecc{s}^{j}$. Then, by the definition of our algorithm (see~\cref{main_alg}), she
has no $p$-move in state $\vecc{s}^{j}$. Applying this along with
\cref{relate_fix_end}, we see that
\begin{align*}
\frac{\cu{\vecc{s}^{m-1}}}{\cu{\vecc{s}_{-u}^{m-1}, s'_u}} 
\leq \frac{1+3p^{-1}}{1-2p^{-1}}  
\frac{\cu{\vecc{s}^{j}}}{\cu{\vecc{s}_{-u}^{j}, s'_u}} 
\leq \frac{1+3p^{-1}}{1-2p^{-1}} \cdot  p
= d^{d + o(d)},
\end{align*}
where the last step holds due to \eqref{eq:pdef} and observing that, since $p\geq 160$ (by~\eqref{eq:pdef} and $d\geq 1$), it must be that $\frac{1+3p^{-1}}{1-2p^{-1}}\leq 163/158\approx 1.032$.
\end{proof}

\paragraph*{Acknowledgements} 
We are grateful to the anonymous reviewers of the journal version of this paper for their valuable
feedback. We also thank Martin Gairing and George Christodoulou for useful
discussions.

\bibliography{PoAapprox}

\appendix

\section*{Appendix}

\section{Proof of \texorpdfstring{\cref{mainlemma}}{Lemma~2}} 
\label{app:poa}
Our proof of~\cref{mainlemma} follows along the lines
of Aland et al.~\cite[Section~5.3]{Aland2011}. Still, care needs to be taken in order to
incorporate correctly the added approximation factor $\rho$ and, to a smaller
degree, parameter $z$ that essentially corresponds to the deviations of subsets of
players in our upper bound result in~\cref{partialpoa}.

We will first need a technical lemma, which is a straightforward generalization
of \cite[Lemma 5.2]{Aland2011}:

\begin{lemma} \label{maximumlemma} Fix any $\mu\geq 0$, $\rho,d>0$ and define the
function $g(x) = \rho (x+1)^d - \mu
\cdot x^{d+1}$. Then $g$ has exactly one local maximum in $\R_{\geq 0}$. More
precisely, there exists a unique $\xi\in\R_{\geq 0}$ such that $g$ is strictly
increasing in $[0,\xi)$ and strictly decreasing in $(\xi, \infty)$.
\end{lemma}

\begin{proof}
For the derivative of $g$ we have that, for any $x>0$:
\begin{equation} \label{maximumeq}
g'(x) = 0 
\quad\ifif\quad
\rho d (x+1)^{d-1} - \mu(d+1)x^d = 0
\quad\ifif\quad
\frac{(x+1)^{d-1}}{x^d} = \frac{\mu (d+1)}{\rho d}.
\end{equation}
Now define the function $h(x)=\frac{(x+1)^{d-1}}{x^d}$ and observe that $h'(x)=-
\frac{(x+1)^{d-2} \cdot (x+d)}{x^{d+1}} < 0$ for all $x>0$. Additionally,
$\lim_{x\to 0^+}h(x)=\infty$ and $\lim_{x\to \infty} h(x)=0$. Thus, since
$c:=\frac{\mu (d+1)}{\rho d}\geq 0$, there exists a \emph{unique} $\xi\in[0,\infty)$
such that $h(\xi)=c$; furthermore, $h(x)>c$ for $x\in[0,\xi)$ and $h(x)<c$ for
$x>\xi$. Given the derivation in \eqref{maximumeq}, this is enough to conclude the
proof of the lemma.
\end{proof}
\medskip

We are now ready to prove~\cref{mainlemma}.
First, examine the critical condition 
$$ y\cdot f(z + x + y) \leq \lambda y \cdot f(z + y) + \mu x \cdot f(z + x) 
\qquad \forall x, y, z \geq 0,\;\forall f \in \mathfrak{P}_d$$ 
of the set whose infimum defines $B$ in the statement of our lemma. 
Note that
for any $z \geq 0$, if we define $\bar{f}(t) := f(z + t)$ for all $t \in
\R$, then $\bar{f}$ is again a polynomial with nonnegative coefficients.
Therefore, the condition above is equivalent to the simpler one: 
$$y \cdot f(x + y) \leq \lambda y \cdot f(y) + \mu x \cdot
f(x)
\qquad \forall x, y, z \geq 0,\;\forall f \in \mathfrak{P}_d.$$ 
From now on, performing a consecutive transformation of this condition in exactly the
same way as in the proof of~\cite[Lemma~5.5]{Aland2011}, we can eventually derive that
$$
B= \inf_{\lambda \in \R, \, \mu \in (0, \frac{1}{\rho})} \left\{
\left. \frac{\lambda \rho}{1 - \mu \rho} \right| \forall x \geq 0 : \; \lambda \geq
(x + 1)^{d} - \mu x^{d+1} \right\}.
$$ 

Now, taking into consideration that $\rho,1-\mu\rho>0$ and $\mu\rho\in(0,1)$, we can
finally do the following transformations:
\begin{align}
B &= \inf_{\lambda \in \R, \, \mu \in (0, \frac{1}{\rho})} \left\{
\left. \frac{\lambda \rho}{1 - \mu \rho} \right| \forall x \geq 0 : \; \frac{\lambda \rho}{1 - \mu \rho} \geq
\frac{\rho (x + 1)^{d} - \mu \rho x^{d+1}}{1 - \mu \rho} \right\} \notag\\
		&= \inf_{\mu \in (0, \frac{1}{\rho})} \left\{ \max_{x \geq 0} \left\{\frac{\rho (x + 1)^{d} - \mu \rho x^{d+1}}{1 - \mu \rho} \right\} \right\}\notag\\
		&= \inf_{\mu \in (0, 1)} \left\{ \max_{x \geq 0} \left\{\frac{\rho (x + 1)^{d} - \mu  x^{d+1}}{1 - \mu} \right\} \right\} \label{eq:B_formula_helper}
\end{align}

Using expression~\eqref{eq:B_formula_helper} for $B$, we will next show that indeed
$B=\grat{d,\rho}^{d+1}$. First, for the $B \geq \grat{d,\rho}^{d+1}$ part, it is
enough to observe that for any $\mu\in(0,1)$:
$$
\max_{x \geq 0} \left\{\frac{\rho (x + 1)^{d} - \mu x^{d+1}}{1 - \mu} \right\} \geq \frac{\rho (\grat{d,\rho} + 1)^{d} - \mu \grat{d,\rho}^{d+1}}{1 - \mu} = \grat{d,\rho}^{d+1}.
$$

For the $B\leq\grat{d,\rho}^{d+1}$ part now, first define the constant
$$
\hat{\mu} := \frac{\rho d (\grat{d,\rho}+1)^{d-1}}{(d+1) \grat{d,\rho}^d}.
$$ 
Using the fundamental property of the generalized golden ratio that $\grat{d,\rho}^{d+1}=\rho(\grat{d,\rho}+1)^d$, we can now see that $\hat{\mu}$ satisfies the following properties:
$$
\hat{\mu} = \frac{d \grat{d,\rho}^{d+1}}{(d+1) (\grat{d,\rho}+1)
\grat{d,\rho}^d} = \frac{d \grat{d,\rho}}{(d+1) (\grat{d,\rho}+1)} \in (0, 1)
$$
and
$$
\frac{\rho (\grat{d, \rho} + 1)^{d} - \hat{\mu} \grat{d, \rho}^{d+1}}{1 - \hat{\mu}}
= \grat{d, \rho}^{d+1}
$$
Additionally, notice that by~\cref{maximumlemma} we thus know that quantity $\rho (x + 1)^{d} - \hat{\mu} x^{d+1}$ is maximized for $x$ being the \emph{unique} solution of equation~\eqref{maximumeq}, which in our case is equivalent to
$$
\frac{(x+1)^{d-1}}{x^d} = \frac{\hat\mu (d+1)}{\rho d}
=  \frac{(\grat{d,\rho}+1)^{d-1}}{ \grat{d,\rho}^d}
\quad\ifif\quad 
x= \grat{d, \rho}.
$$  

So, we can finally upper-bound $B$ via~\eqref{eq:B_formula_helper} by 
$$
B\leq \max_{x \geq 0}\left\{\frac{\rho (x + 1)^{d} - \hat{\mu} x^{d+1}}{1 - \hat{\mu}} \right\} =
\frac{\rho (\grat{d, \rho} + 1)^{d} - \hat{\mu} \grat{d, \rho}^{d+1}}{1 - \hat{\mu}}
= \grat{d, \rho}^{d+1}.
$$

\section{Remaining Proof of~\texorpdfstring{\cref{thm}}{Theorem~3}}
\label{app:proof_key}

In this section we show that the profiles $\bar{\vecc{s}}^{i-1}$ and
$\bar{\vecc{s}}^{i}$ of the game $\bar\Gamma$ defined in the proof of~\cref{thm},
indeed have the desired properties. Namely, $\bar{\vecc{s}}^{i-1}$ is an $(\alpha +
p^{-1})$-equilibrium for all players in $R_i$, and furthermore,
\begin{equation}
\label{eq:helper_key_1}
\pbar{R_i}{\bar{\vecc{s}}^{i}}{}
\leq \p{R_i}{\vecc{s}^{i}}{} + nb_i.
\end{equation}
Our analysis is essentially the same as in the proof
of Caragiannis et al.~\cite[Lemma~4.3]{Caragiannis2015a}; we still present our derivation below, for
completeness.

First, we need to show that no player in $R_i=X_i\union Y_i$ has an $(\alpha +
p^{-1})$-move at state $\bar{\vecc{s}}^{i-1}$ of game $\bar\Gamma$. Recall that all
players use the same strategies under both $\bar{\vecc{s}}^{i-1}$ and
$\vecc{s}^{i-1}$ and, furthermore, the cost functions of the resources in these
strategies are the same between the two games $\bar\Gamma$ and $\Gamma$. Also,
observe that the cost of all players in $R_i$ at $\vecc{s}^{i-1}$ is strictly less
than $b_{i-1}$, otherwise they would have been ``fixed'' at the end of the
$(i-1)$-th phase (see \cref{algo_nfixed} of \cref{main_alg}), contradicting the
definition of $R_i$ containing players that perform improving moves during the
$i$-th phase.

Consider first the players in $X_i$, whose cost by definition is at least $b_i$ at
$\vecc{s}^{i-1}$. By the analysis above, their costs actually have to lie within
$[b_i,b_{i-1})$ at the end of the $(i-1)$-th phase. Furthermore, their strategy sets
(by the definition of $\bar{\Gamma}$) are exactly the same in $\bar{\Gamma}$ and
$\Gamma$, and thus, any $(\alpha + p^{-1})$-move from $\bar{\vecc{s}}^{i-1}$ in
$\bar{\Gamma}$ would give rise to an $(\alpha + p^{-1})$-move from $\vecc{s}^{i-1}$
in $\Gamma$; such a move cannot exist, however, or otherwise, by the definition of our
algorithm, there would be still $(\alpha + p^{-1})$-moves left for the players in $X_i$ at
the end of phase $(i+1)$ (see \cref{algo_improvements} of \cref{main_alg}).

Consider now the remaining players in $Y_i$, whose cost is less than $b_i$. A
unilateral deviation from $\bar{\vecc{s}}^{i-1}$ of any such player $u \in Y_i$,
currently playing $s_u$, would cause her to use resources $s'_u \cup \{e_u\}$ for
some $s'_u \neq s_u$. This will result in her experiencing a cost of at least $w_u
c_{e_u}(w_u)=w_u\cdot \frac{b_i}{w_u} = b_i$, which is strictly worse than her
current cost.

Now we move on to prove~\eqref{eq:helper_key_1}. Recall that all players in
$N\setminus Y_i$ use the same strategies in $\bar{\vecc{s}}^{i}$ and $\vecc{s}^{i}$.
Also, each resource $e_u\in \bar{E}\setminus E$ can be used only by player $u\in
Y_i$, and thus $x_{N \setminus R_i, e_u}(\bar{\vecc{s}}^{i}) = 0$ and $x_{R_i,
e_u}(\bar{\vecc{s}}^{i}) \leq w_{u}$, and furthermore, from the definition of our
potential in \eqref{potentialDefine}, it is $\phi_{e_u}(x) = \frac{b_i}{w_u}x$.
Putting all the above together and using the definition of the partial
potential~\eqref{eq:partial_pot_def}, we can derive that
\begin{align*}
\pbar{R_i}{\bar{\vecc{s}}^{i}}{}
= & \sum_{e \in E'} 
\left[\pe{\xe{N \setminus R_i}{\bar{\vecc{s}}^{i}} + \xe{R_i}{\bar{\vecc{s}}^{i}}}
- \pe{\xe{N \setminus R_i}{\bar{\vecc{s}}^{i}}}\right] \\ 
= & \sum_{e \in E} 
\left[\pe{\xe{N \setminus R_i}{\bar{\vecc{s}}^{i}} + \xe{R_i}{\bar{\vecc{s}}^{i}}}
- \pe{\xe{N \setminus R_i}{\bar{\vecc{s}}^{i}}}\right] \\
& + \sum_{u \in Y_i} 
\left[\phi_{e_u}(x_{N \setminus R_i,e_u}(\bar{\vecc{s}}^{i}) + x_{R_i,e_u}(\bar{\vecc{s}}^{i})) 
- \phi_{e_u}(x_{N \setminus R_i,e_u}(\bar{\vecc{s}}^{i}))\right]\\
= & \;\; \varPhi_{R_i}({\vecc s}^{i}) + \sum_{u \in Y_i} [\phi_{e_u}(x_{R_i,e_u}(\bar{\vecc{s}}^{i})) 
- \phi_{e_u}(0)]\\
\leq & \;\; \varPhi_{R_i}({\vecc s}^{i}) + \sum_{u \in Y_i} \phi_{e_u}(w_u)\\
\leq & \;\; \varPhi_{R_i}({\vecc s}^{i}) + n b_i.
\end{align*}

\section{Proof of \texorpdfstring{\cref{lemma2}}{Lemma~10}}
\label{app:proof_lemma2}

We will first need a series of technical algebraic lemmas.

\begin{lemma} \label{concav}
For any $\psi \in (0, 1]$ and all $x > 0$: 
$$(1 + x)^{\psi} - 1 \geq \psi x (1 + x)^{\psi - 1}.$$
\end{lemma}

\begin{proof}
Set $\alpha\gets\psi$ and $z\gets x+1$ in~\cite[Claim~2.2]{Caragiannis2015a}.
\end{proof}

\begin{lemma} \label{cl1}
Consider reals $\varepsilon,p>0$ and $m\geq 1$, such that $(1 + \varepsilon)^{m} = 1 + p^{-1}$. Then, $\varepsilon^{-1} \leq m(1 + p)$.
\end{lemma}
\begin{proof}
First, rearranging our assumption that $(1 + \varepsilon)^{m} = 1 + p^{-1}$, we get that
$$
\varepsilon = \left(1 + \frac{1}{p}\right)^{1/m} - 1.
$$
Thus, applying \cref{concav} with $x \gets \frac{1}{p}$ and $\psi \gets \frac{1}{m}$ we can derive that:
$$
\varepsilon 
\geq \frac{1}{pm} \left(1 + \frac{1}{p}\right)^{1/m-1}
=\frac{1}{m(1+p)} \left(1 + \frac{1}{p}\right)^{1/m}
\geq \frac{1}{m(1+p)},
$$
which concludes the proof.
\end{proof}

The following lemma builds the algebraic foundation of the derivation of~\cref{lemma2}.

\begin{lemma} \label{lemmacomb}
Fix a positive integer $d$.
For any polynomial $f\in \mathfrak{P}_d$ and any $\varepsilon > 0$, 
it holds that
$$
y f(z + x + y) \leq (1+\varepsilon) y f(z + x' + y) 
+ \left(1+\frac{1}{\varepsilon} \right)^d d^d x f(z + x + y'),
$$
for any $x, x', y, y', z \geq 0$.
\end{lemma}
\begin{proof}
Since $f$ is a nondecreasing function, it is enough to show that (setting $x'\gets 0$, $y'\gets 0$)
\begin{align*}
y f(z + x + y) \leq (1+\varepsilon)  y  f(z + y) 
+ \left(1+\frac{1}{\varepsilon} \right)^d d^d  x  f(z + x)
\end{align*}
and by performing the transformation $f(t) \gets f(z + t)$ (which preserves the property that $f\in\mathfrak{P}_d$), it is actually enough to show that
\begin{equation} \label{newClaim}
y f(x + y) \leq (1+\varepsilon) y f(y) 
+ \left(1+\frac{1}{\varepsilon} \right)^d d^d x f(x).
\end{equation}
Looking at the structure of the above inequality, and taking into consideration that
$f$ is a linear combination (with nonnegative coefficients) of monomials of degree
at most $d$, we can deduce that it is enough to show that \eqref{newClaim} holds for
any $f(t)=t^\nu$ with $\nu=0,1,\dots,d$, that is,
\begin{align*}
y(x+y)^{\nu} \leq (1+\varepsilon) y^{\nu+1} 
+ \left(1+\frac{1}{\varepsilon} \right)^d d^d x^{\nu+1}.
\end{align*}
Dividing by $x^{\nu+1}$ 
(if $x = 0$, the inequality holds trivially)
and making a change of variables $z\gets \frac{y}{x}$,
this can be rewritten as 
$$ 
z(z+1)^\nu - (1+\varepsilon)z^{\nu+1} 
\leq \left(1+\frac{1}{\varepsilon} \right)^d d^d.
$$
It is not difficult to check that the above inequality indeed holds for the extreme
cases of $z=0$ or $\nu=0$, so from now on let's assume that $z>0$ and $\nu$ is a
positive integer. Now, taking into consideration that $\nu\leq d$ and that $z\leq
z+1$, we can finally derive that it is enough to show that
\begin{equation}
\label{eq:helper_g_nu}
g(z):= (z+1)^{\nu+1} - (1+\varepsilon)z^{\nu+1} \leq \left (1+\frac{1}{\varepsilon} \right)^\nu \nu^{\nu}
\end{equation}
for all $z>0$, where $\nu$ is a positive integer.

The derivative of $g$ defined in~\eqref{eq:helper_g_nu} is $g'(z) = (\nu+1)[(z +
1)^{\nu} - (1+\varepsilon) z^{\nu}]$, and so $\lim_{z\to 0}g'(z)=\nu+1>0$,
$\lim_{z\to \infty}g'(z)=-\infty$ (because $\varepsilon >0$) and
\begin{equation}
\label{eq:props_zbar}
g'(z) = 0 
\ifif \frac{(z + 1)^{\nu}}{z^{\nu}} = 1 + \varepsilon
\ifif z=\bar z:= \frac{1}{(1 + \varepsilon)^{1/\nu} - 1}.
\end{equation}
Thus, the maximum of $g$ in $(0,\infty)$ is attained at $\bar z$. Also, by
applying~\cref{concav} with $\psi = 1/\nu$ and $x = \varepsilon$, we get that $(1 +
\varepsilon)^{1/\nu} - 1 \geq \frac{1}{\nu} \varepsilon (1 + \varepsilon)^{1/\nu -
1}$, and so we can bound $\bar z$ in~\eqref{eq:props_zbar} by
$$
\bar z \leq \nu \frac{1}{\varepsilon} (1 + \varepsilon)^{1-1/\nu}.
$$
Using the above, together with the properties of $\bar z$ from \eqref{eq:props_zbar}
we can finally bound
\begin{align*}
g(z) &\leq (\bar z+1)^{\nu+1} - (1+\varepsilon){\bar z}^{\nu+1}\\
&= (\bar z+1)\cdot (1+\varepsilon){\bar z}^{\nu} - (1+\varepsilon){\bar z}^{\nu+1}\\
&= (1+\varepsilon){\bar z}^{\nu}\\
& \leq (1+\varepsilon) \nu^\nu \frac{1}{\varepsilon^\nu} (1 + \varepsilon)^{\nu-1}\\
&= \left (1+\frac{1}{\varepsilon} \right)^\nu \nu^{\nu},
\end{align*}
proving \eqref{eq:helper_g_nu} and concluding the proof of our lemma.
\end{proof}

\medskip\medskip

We are finally ready to prove~\cref{lemma2}.
Throughout this proof we will denote $Q = R_i \cup \{u\}$. Then we can write
$$
\cu{\vecc{s}^{i}} = \sum_{e\in E} \xe{u}{\vecc{s}^{i}} \cdot
\ce{\xe{N \setminus Q}{\vecc{s}^{i}} 
+ \xe{R_i}{\vecc{s}^{i}} + \xe{u}{\vecc{s}^{i}}}
$$
and
$$
\cu{\vecc{s}^{i-1}} = \sum_{e\in E} \xe{u}{\vecc{s}^{i}} \cdot 
\ce{\xe{N \setminus Q}{\vecc{s}^{i}} + \xe{R_i}{\vecc{s}^{i-1}} + \xe{u}{\vecc{s}^{i}}},
$$
where for the second equality we used that, for any resource $e$, it is
$\xe{N \setminus Q}{\vecc{s}^{i}} = \xe{N \setminus Q}{\vecc{s}^{i-1}}$
and
$\xe{u}{\vecc{s}^{i}} = \xe{u}{\vecc{s}^{i-1}}$;
this is a consequence of the fact that only players in $R_i$ move during phase $i$, 
i.e., between states $\vecc{s}^{i-1}$ and $\vecc{s}^{i}$. 
This is also exactly the reason why the partial potential of $R_i$ can only
decrease during phase $i$ (due to \cref{potentialProp} and the fact that players
perform at least $\alpha$-improvements). Thus
$
\p{R_i}{\vecc{s}^{i-1}}{} 
\geq \p{R_i}{\vecc{s}^{i}}{}
= \p{}{\vecc{s}^{i}}{} - \p{}{\vecc{s}^{i}}{N \setminus R_i}
$
and so we can also write
\begin{align*}
\p{R_i}{\vecc{s}^{i-1}}{} 
 & \geq \sum_{e\in E} \left[\pe{\xe{N \setminus Q}{\vecc{s}^{i}} + 
\xe{R_i}{\vecc{s}^{i}} + \xe{u}{\vecc{s}^{i}}} 
- \pe{\xe{N \setminus Q}{\vecc{s}^{i}} + \xe{u}{\vecc{s}^{i}}}\right]\\
 & \geq \sum_{e\in E} \xe{R_i}{\vecc{s}^{i}} \cdot
\ce{\xe{N \setminus Q}{\vecc{s}^{i}} + \xe{R_i}{\vecc{s}^{i}} + \xe{u}{\vecc{s}^{i}}},
\end{align*}
the last inequality being due to~\cref{thelemma}.

Observing the above expressions, we see that in order to prove~\eqref{eq:lemma2_1}
it is enough to show that, for any resource $e$, it is
\begin{align*}
\xe{u}{\vecc{s}^{i}}
\ce{\xe{N \setminus Q}{\vecc{s}^{i}} &+ \xe{R_i}{\vecc{s}^{i}} + \xe{u}{\vecc{s}^{i}}} \\
&\leq   (1+\varepsilon) \xe{u}{\vecc{s}^{i-1}}  
\ce{\xe{N \setminus Q}{\vecc{s}^{i-1}} + \xe{R_i}{\vecc{s}^{i-1}} + \xe{u}{\vecc{s}^{i-1}}} \\
&\qquad +  \xi_{\varepsilon}   
\xe{R_i}{\vecc{s}^{i}}  
\ce{\xe{N \setminus Q}{\vecc{s}^{i}} + \xe{R_i}{\vecc{s}^{i}} + \xe{u}{\vecc{s}^{i}}}. 
\end{align*}
Substituting, for simplicity, $y \gets \xe{u}{\vecc{s}^{i}}$, $x \gets \xe{R_i}{\vecc{s}^{i}}$, $x' \gets \xe{R_i}{\vecc{s}^{i-1}}$ and $z \gets \xe{N \setminus Q}{\vecc{s}^{i}}$, this becomes 
$$
y c_e(z + x + y) \leq (1+\varepsilon) y c_e(z + x' + y) 
+ \xi_{\varepsilon} x c_e(z + x + y),
$$
which indeed holds, due to \cref{lemmacomb}.

We now move to proving \eqref{eq:lemma2_2}, which can be equivalently rewritten as
\begin{equation}
\label{eq:lemma2_2_new}
\cu{\vecc{s}_{-u}^{i-1}, s'_u} 
\leq (1+\varepsilon) \cu{\vecc{s}_{-u}^{i}, s'_u} 
+ \xi_{\varepsilon} \p{R_i}{\vecc{s}^{i-1}}{}.
\end{equation}
We have that
\begin{align*}
\cu{\vecc{s}_{-u}^{i-1}, s'_u} 
&=
\sum_{e\in E} \xe{u}{\vecc{s}_{-u}^{i-1}, s'_u} \cdot 
   \ce{\xe{N \setminus Q}{\vecc{s}_{-u}^{i-1}, s'_u} 
   + \xe{R_i}{\vecc{s}_{-u}^{i-1}, s'_u} + \xe{u}{\vecc{s}_{-u}^{i-1}, s'_u}}\\
&= \sum_{e\in E} \xe{u}{\vecc{s}_{-u}^{i}, s'_u} \cdot 
   \ce{\xe{N \setminus Q}{\vecc{s}_{-u}^{i-1}} 
   + \xe{R_i}{\vecc{s}_{-u}^{i-1}} + \xe{u}{\vecc{s}_{-u}^{i}, s'_u}} 
\end{align*}
and
\begin{align*}
\cu{\vecc{s}_{-u}^{i}, s'_u} 
&=
\sum_{e\in E} \xe{u}{\vecc{s}_{-u}^{i}, s'_u} \cdot 
\ce{\xe{N \setminus Q}{\vecc{s}_{-u}^{i}, s'_u} 
+ \xe{R_i}{\vecc{s}_{-u}^{i}, s'_u} + \xe{u}{\vecc{s}_{-u}^{i}, s'_u}}\\
&=
\sum_{e\in E} \xe{u}{\vecc{s}_{-u}^{i}, s'_u} \cdot 
\ce{\xe{N \setminus Q}{\vecc{s}_{-u}^{i-1}} 
+ \xe{R_i}{\vecc{s}_{-u}^{i}, s'_u} + \xe{u}{\vecc{s}_{-u}^{i}, s'_u}}.
\end{align*}
For the simplifications above we used the equalities
\begin{align*}
\xe{u}{\vecc{s}_{-u}^{i-1}, s'_u} &= \xe{u}{\vecc{s}_{-u}^{i}, s'_u}\\
\xe{R_i}{\vecc{s}_{-u}^{i-1}, s'_u} &= \xe{R_i}{\vecc{s}^{i-1}}\\
\xe{N \setminus Q}{\vecc{s}_{-u}^{i-1}, s'_u} &= \xe{N \setminus Q}{\vecc{s}_{-u}^{i}, s'_u} = \xe{N \setminus Q}{\vecc{s}^{i-1}},
\end{align*}
which are consequences of the fact that all players 
except those in $R_i$ have the same strategies in 
$(\vecc{s}_{-u}^{i}, s'_u)$ and $(\vecc{s}_{-u}^{i-1}, s'_u)$, and 
all players but $u$ have the same strategies in state 
$\vecc{s}^{i-1}$ and $(\vecc{s}_{-u}^{i-1}, s'_u)$. 

Thus, in an analogous way in which we proved \eqref{eq:lemma2_1} before, we can see that in order to prove \eqref{eq:lemma2_2_new} it is enough to show that for any resource $e$ it is
$$
y c_e(z + x + y) \leq (1+\varepsilon) y c_e(z + x' + y) 
+ \xi_{\varepsilon} x c_e(z + x + y'),
$$
where we have used the following substitutions: $x \gets \xe{R_i}{\vecc{s}^{i-1}}$, 
$x' \gets \xe{R_i}{\vecc{s}_{-u}^{i}, s'_{u}}$,  
$y \gets \xe{u}{\vecc{s}_{-u}^{i-1}, s'_u}$,  
$y' \gets \xe{u}{{\vecc{s}}^{i-1}}$ and 
$z \gets \xe{N \setminus Q}{\vecc{s}^{i}}$. The above is indeed true, again due to \cref{lemmacomb}.

\end{document}